\documentclass[preprint, 12pt]{elsarticle}

\usepackage{amssymb}
\usepackage{amsmath}
\usepackage{amsthm}
\usepackage{bm}
\usepackage{booktabs}
\usepackage{multirow}
\usepackage{rotating}
\usepackage{pdflscape}
\usepackage{float}   

\usepackage{siunitx}
\newcommand{\diag}{\operatorname{diag}}
\theoremstyle{plain}
\newtheorem{theorem}{Theorem}
\newtheorem{lemma}{Lemma}

\theoremstyle{definition}

\newtheorem{assumption}{Assumption}

\theoremstyle{remark}
\newtheorem{remark}[theorem]{Remark}

\journal{Spatial Statistics}

\begin{document}

\begin{frontmatter}

\title{Sparse\textendash Smooth Spatially Varying Coefficient Quantile Regression}

\author{HOU Jian}
\affiliation{organization={Center for Applied Statistics, School of Statistics},
            addressline={Renmin University of China}, 
            city={Beijing},
            postcode={100872}, 
            state={Beijing},
            country={China}}

\author{MENG Tan}
\affiliation{organization={Center for Applied Statistics, School of Statistics},
            addressline={Renmin University of China}, 
            city={Beijing},
            postcode={100872}, 
            state={Beijing},
            country={China}}

\author{TIAN Maozai}
\affiliation{organization={Center for Applied Statistics, School of Statistics},
            addressline={Renmin University of China}, 
            city={Beijing},
            postcode={100872}, 
            state={Beijing},
            country={China}}

\begin{abstract}
    Modeling spatially varying relationships in the presence of heavy-tailed responses and irregular sampling designs presents a significant challenge in geospatial analysis. In this paper, we propose a Sparse--Smooth Spatially Varying Coefficient Quantile Regression (SS--SVCQR) framework that simultaneously achieves robust estimation, structural interpretability, and computational scalability. Our approach decomposes each coefficient function into a global baseline and a spatial deviation component. By applying a group $L_2$ penalty, the model automatically selects between globally constant and spatially varying specifications, offering a parsimonious alternative to standard Geographically Weighted Regression (GWR). Conditional on local variation, spatial smoothness is enforced via a normalized graph Laplacian regularization, which adapts to varying sampling densities and admits a Bayesian interpretation as an intrinsic Gaussian Markov Random Field (GMRF) prior. We develop efficient convex optimization algorithms based on ADMM and smoothed proximal-gradient methods to handle the non-differentiable check loss. Theoretically, we establish the root-$n$ consistency and asymptotic normality of the global estimators, as well as the selection consistency of the global-vs-local structure under mild conditions. Extensive simulations and a real-world housing application show that SS--SVCQR provides more accurate quantile predictions and more interpretable global–local structures than existing kernel- and spline-based alternatives, especially in the presence of heavy-tailed noise.
\end{abstract}

\begin{graphicalabstract}
    \includegraphics{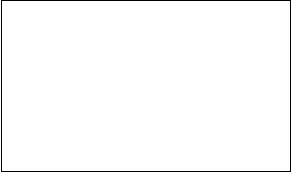}
\end{graphicalabstract}

\begin{highlights}
    \item Parsimonious global--local decomposition distinguishes constant effects from spatial deviations via adaptive group selection.
    \item Normalized Graph-Laplacian smoothing adapts to irregular sampling densities while ensuring geometric fidelity.
    \item Convex formulation solved by efficient ADMM and Moreau-smoothed proximal-gradient algorithms.
    \item Establishes oracle properties: selection consistency for local deviations and root-$n$ normality for global coefficients.
    \item Outperforms kernel-based methods in recovering heterogeneous patterns under heavy-tailed distributions.
\end{highlights}


\end{frontmatter}


\section{Introduction}\label{sec:introduction}

Spatial nonstationarity—the phenomenon where relationships between a response variable and its predictors vary across geographical locations—poses a fundamental challenge in regression analysis. Ignoring this heterogeneity by fitting a global constant-coefficient model can lead to biased estimates and misleading inference. To address this, the literature on Spatially Varying Coefficient (SVC) models has flourished. The most prominent approach is Geographically Weighted Regression (GWR) \cite{Brunsdon1996, Fotheringham2002}, which fits local weighted regressions at each target location. While effective, standard GWR often suffers from multicollinearity and bandwidth selection sensitivity. Recent extensions have sought to mitigate these issues through multiscale formulations \cite{FotheringhamYangKang2017, OshanFotheringham2019} and mixed GWR (MGWR) that allows some coefficients to remain fixed \cite{Mei2006, Geniaux2017}. Alternative SVC frameworks include Bayesian hierarchical models based on Gaussian Processes (GPs) \cite{Gelfand2003SVC, Banerjee2008} and frequentist approaches using penalized splines \cite{Eilers2010, Xiao2013}. However, GP-based methods often face cubic computational complexity, limiting their applicability to large-scale geospatial datasets.

Parallel to mean regression, Quantile Regression (QR) \cite{KoenkerBassett1978} has become an indispensable tool for characterizing the entire conditional distribution. In spatial contexts, modeling conditional quantiles is particularly valuable because spatial heterogeneity often manifests asymmetrically—affecting the tails of the distribution differently than the center \cite{Reich2012STQR}. Although Spatially Varying Coefficient Quantile Regression (SVC-QR) offers a powerful lens for such analysis, existing methods face dual challenges: interpretability and scalability. Many nonparametric SVC-QR estimators \cite{Hallin2004, Lum2012} yield coefficients that vary everywhere, failing to distinguish between truly local effects and those that are effectively global. This lack of parsimony complicates interpretation. While variable selection techniques like the Lasso have been adapted for spatial models \cite{Wang2008, Cai2013}, scalable convex formulations that simultaneously deliver interpretable global-vs-local selection and graph-based smoothness remain relatively limited. Recent developments in "Network Lasso" and graph-trend filtering \cite{Hallac2015, Wang2016} suggest that graph-based penalties offer a computationally efficient alternative to kernel or spline methods, yet their integration with quantile loss in a spatial setting, together with an explicit mechanism for global–vs–local selection, remains underexplored.

In this paper, we propose a penalized estimator that addresses these gaps by integrating graph Laplacian regularization with group-sparse penalization. Our approach, termed SS-SVCQR (Sparse-Smooth Spatially Varying Coefficient Quantile Regression), is designed with two key desiderata. First, we employ a decomposition strategy that separates each coefficient into a global baseline and a spatial deviation. A group $L_2$ penalty \cite{YuanLin2006} is applied to the deviation vector, shrinking it to zero when the covariate effect is globally constant. This provides an automated, interpretable decision mechanism for global-vs-local specification, a feature often absent in kernel-weighted local QR. Second, conditioned on being local, we regularize the deviations using a normalized graph Laplacian constructed from the sampling locations \cite{Smola2003, Belkin2003}. This encourages spatial continuity and allows the model to adapt to irregular sampling designs without the heavy computational burden of dense covariance matrices used in GPs. The resulting formulation is convex and compatible with large-scale solvers, including Alternating Direction Method of Multipliers (ADMM) schemes tailored to the non-differentiable check loss.

On the theoretical front, we establish that our estimator enjoys desirable asymptotic properties. We show that the global coefficients are root-$n$ consistent and asymptotically normal, while the deviation fields converge in mean square at a rate that balances Laplacian bias and stochastic variability. Furthermore, we prove that the selection of globally constant effects is consistent under a groupwise beta-min condition. Our theoretical framework accommodates random sampling locations filling a compact set and sparse graphs with uniformly bounded average degrees.

The remainder of this paper is organized as follows. Section \ref{sec:model} introduces the proposed model decomposition and the graph Laplacian regularization framework. Section \ref{sec:estimation} details the optimization algorithms, including an efficient ADMM implementation and tuning parameter selection via spatially blocked cross-validation. Section \ref{sec:theory} provides the asymptotic analysis. In Section \ref{sec:simulation}, we evaluate the finite-sample performance of our method against competing GWR and spline-based approaches through extensive simulations. Section \ref{sec:realdata} applies our method to a real-world analysis of housing prices in the Lucas County, demonstrating the model's ability to uncover heterogeneous pricing mechanisms and distributional tail dynamics. Section \ref{sec:discussion} concludes with a discussion.

\section{Model}\label{sec:model}

We develop a sparse--smooth spatially varying coefficient quantile regression (SS--SVCQR) at a fixed quantile $\tau\in(0, 1)$. This framework fuses the robustness of quantile regression with the flexibility of graph-based regularization to capture heterogeneous spatial dependencies.

Scalars are denoted by italic lower case, vectors by bold lower case, and matrices by bold upper case. For a vector $\bm v$, $\|\bm v\|_2$ denotes the Euclidean norm, and $\bm 1$ is an all-ones vector of conformable length.

\subsection{Data and Decomposition Structure}

We observe a spatial dataset $\{(Y_i,\bm Z_i,\bm X_i,u_i): i=1,\dots,n\}$, where $Y_i\in\mathbb R$ is the response, $\bm Z_i\in\mathbb R^{q}$ represents covariates with purely global effects, $\bm X_i\in\mathbb R^{p}$ contains predictors with potentially spatially varying effects, and $u_i\in\mathcal U\subset\mathbb R^2$ denotes the spatial location. Let $Q_\tau(Y_i\mid \bm Z_i,\bm X_i,u_i)$ be the conditional $\tau$-quantile. We assume the data generating process follows a varying-coefficient structure:
\begin{equation}\label{eq:model}
    Q_\tau\{Y_i\mid \bm Z_i, \bm X_i, u_i\}
    = \bm Z_i^\top\bm\alpha + \sum_{j=1}^{p} X_{ij} \beta_j(u_i),
    \quad
    i = 1, \dots, n,
\end{equation}
where $\bm\alpha\in\mathbb R^{q}$ collects the fixed global coefficients, and $\beta_j(\cdot)$ are unknown functions characterizing the spatial heterogeneity. To simultaneously perform variable selection between global and local specifications, we decompose each varying coefficient into a global baseline level and a location-specific spatial deviation:
\begin{equation}\label{eq:decomp}
    \beta_j(u) = \beta_{G,j} + \delta_j(u),
    \quad j=1, \dots, p.
\end{equation}
Here, $\bm\beta_G=(\beta_{G,1},\dots,\beta_{G,p})^\top\in\mathbb R^{p}$ captures the domain-wide average effect, while $\delta_j:\mathcal U\to\mathbb R$ captures local departures. This additive form allows the model to adaptively shrink $\delta_j(\cdot)$ to zero when the effect of predictor $j$ is homogeneous, reducing the model to a standard global quantile regression.

The model is anchored at level $\tau$ by the quantile exogeneity condition:
\begin{equation}\label{eq:qexog}
    \Pr \big\{Y_i \le \bm Z_i^\top\bm\alpha_0 + \sum_{j=1}^{p} X_{ij} \beta_{0,j}(u_i) \big| \bm Z_i,\bm X_i,u_i\big\} = \tau, \
    \text{almost surely for each } i.
\end{equation}
Here $\bm\beta_0 = \{\beta_{0, 1}(\cdot), \dots, \beta_{0, p}(\cdot)\}$ is represented by $(\bm\beta_{G0}, \{\delta_{0, j}(\cdot)\}_{j=1}^p)$ with $\beta_{0, j}(u)=\beta_{G0, j}+\delta_{0, j}(u)$.

\subsection{Graph Construction and Spatial Regularization}

Since the locations $\{u_i\}$ may be irregularly distributed, we encode the spatial domain's geometry via a proximity graph. Let $\mathcal{G}=(\mathcal{V}, \mathcal{E})$ be a graph with vertices $\mathcal{V}=\{1,\dots,n\}$ representing the observations. We construct a sparse, symmetric adjacency matrix $\bm A\in\mathbb R^{n\times n}$ using a mutual $k$-nearest-neighbor ($k$-NN) rule combined with a Gaussian kernel. Specifically, for two locations $u_i$ and $u_\ell$, the edge weight is defined as:
\begin{equation}
    A_{i\ell} =
    \begin{cases}
        \exp\left(-\frac{\|u_i - u_\ell\|_2^2}{\sigma^2}\right), & \text{if } u_\ell \in \mathcal{N}_k(u_i) \text{ or } u_i \in \mathcal{N}_k(u_\ell), \\
        0, & \text{otherwise},
    \end{cases}
\end{equation}
where $\mathcal{N}_k(u)$ denotes the set of $k$ nearest neighbors of $u$, and $\sigma$ is a bandwidth parameter. This construction ensures the graph is connected while adapting to local sampling densities.

To enforce spatial smoothness on $\bm\delta_j$, we utilize the spectral properties of the graph. Let $\bm D=\diag(\bm A\bm 1)$ be the degree matrix. We adopt the symmetric normalized Laplacian:
\begin{equation}\label{eq:Lsym}
    \bm L_{\mathrm{sym}} = \bm I - \bm D^{-1/2}\bm A \bm D^{-1/2}.
\end{equation}
We favor $\bm L_{\mathrm{sym}}$ over the unnormalized Laplacian ($\bm D - \bm A$) because its spectrum is bounded in $[0,2]$ and, crucially, it normalizes for node degree. In spatial statistics, sampling locations often exhibit clustering; the normalized Laplacian ensures that the regularization penalty is not artificially inflated in densely sampled regions \cite{Chung1997}.

The spatial roughness of the deviation vector $\bm\delta_j = (\delta_j(u_1), \dots, \delta_j(u_n))^\top$ is quantified by the quadratic form:
\begin{equation}\label{eq:rough}
    \bm\delta_j^\top \bm L_{\mathrm{sym}} \bm\delta_j
    = \frac12 \sum_{(i,\ell)\in\mathcal E} A_{i\ell}\left\{\frac{\delta_j(u_i)}{\sqrt{d_i}} - \frac{\delta_j(u_\ell)}{\sqrt{d_\ell}}\right\}^2,
    \quad d_i=(\bm D)_{ii}.
\end{equation}
This penalty term admits a probabilistic interpretation. It corresponds to placing an intrinsic Gaussian Markov Random Field (GMRF) prior on the spatial deviations, $p(\bm\delta_j) \propto \exp(-\frac{\lambda_2}{2} \bm\delta_j^\top \bm L_{\mathrm{sym}} \bm\delta_j)$. This prior encodes the belief that, conditional on its neighbors, the deviation at a node is normally distributed around the degree-weighted average of its neighbors, thereby enforcing local coherence.

\subsection{Identifiability and Matrix Formulation}

The decomposition in \eqref{eq:decomp} requires a constraint to prevent collinearity between the global intercept $\beta_{G,j}$ and the spatial field $\delta_j(\cdot)$. We impose that local deviations are degree-weight centered componentwise. For a connected component $\mathcal C\subset\{1,\dots,n\}$ with an indicator vector $\bm 1_{\mathcal C}$, the identifiability constraint is:
\begin{equation}\label{eq:ident}
    \bm 1_{\mathcal C}^\top \bm D  \bm\delta_j = 0,
\end{equation}
for all $j=1,\dots,p$. Since the null space of $\bm L_{\mathrm{sym}}$ on a component is spanned by $\bm D^{1/2}\bm 1_{\mathcal C}$, ensuring $\bm D^{1/2}\bm\delta_j$ is orthogonal to this null space is equivalent to \eqref{eq:ident}. Under this constraint, $\beta_{G,j}$ interprets as the \emph{degree-weighted spatial average} of the total coefficient $\beta_j(\cdot)$, while $\delta_j(\cdot)$ captures pure spatial variation.

For estimation, we vectorize the system. Stack the responses $\bm y=(Y_1, \dots, Y_n)^\top$ and design matrices $\bm Z=[\bm Z_1^\top; \cdots; \bm Z_n^\top]\in\mathbb R^{n\times q}$ and $\bm X=[\bm X_1^\top; \dots; \bm X_n^\top]\in\mathbb R^{n\times p}$. To handle the varying coefficients, we define diagonal operators $\bm X_{\odot j}=\diag(X_{1j},\dots,X_{nj})\in\mathbb R^{n\times n}$ for $j=1, \dots, p$. Concatenating the deviations as $\bm\Delta=\big[\bm\delta_1^\top, \dots, \bm\delta_p^\top\big]^\top \in\mathbb R^{np}$, we construct the block operator:
\begin{equation}\label{eq:Xdelta}
    \bm X_\delta = \big[\bm X_{\odot 1}\ \bm X_{\odot 2}\ \cdots\ \bm X_{\odot p}\big]\in\mathbb R^{n\times np}.
\end{equation}
so that $X_\delta \Delta$ stacks all local deviation terms in a linear operator form. The global vector of conditional quantiles is then succinctly expressed as:
\begin{equation}\label{eq:vector}
    \bm q_\tau = \bm Z\bm\alpha + \bm X\bm\beta_G + \bm X_\delta \bm\Delta.
\end{equation}
This linear formulation allows us to leverage efficient convex optimization algorithms to estimate the parameters $\theta=(\bm\alpha,\bm\beta_G,\bm\Delta)$ by minimizing the check loss subject to the penalties defined above.

\section{Estimation and Algorithms}\label{sec:estimation}

We estimate $(\bm\theta_P, \bm\Theta_L)$ with $\bm\theta_P=(\bm\alpha^\top,\bm\beta_G^\top)^\top\in\mathbb R^{q+p}$ and $\bm\Theta_L=\{\bm\delta_j\in\mathbb R^n: j=1,\dots,p\}$, each $\bm\delta_j$ degree-weight centered on every component.

The empirical objective at $\tau$ is
\begin{equation}\label{eq:obj-main}
    \begin{aligned}
        \min_{\bm\theta_P,\bm\Theta_L}  
        \sum_{i=1}^n \rho_\tau \Big(Y_i - \bm Z_i^\top\bm\alpha - \bm X_i^\top\bm\beta_G - \sum_{j=1}^p X_{ij} \delta_j(u_i)\Big)
        +   \lambda_{1} \sum_{j=1}^p w_{j} \|\bm\delta_j\|_2 \\
        +   \lambda_{2} \sum_{j=1}^p \bm\delta_j^\top \bm L_{\mathrm{sym}} \bm\delta_j,
    \end{aligned}
\end{equation}
subject to $\bm 1_{\mathcal C}^\top \bm D \bm\delta_j=0$ for every component $\mathcal C$ and all $j$. Weights $w_j>0$ enable adaptive grouping, e.g.\ $w_j=(\|\tilde{\bm\delta}_j\|_2 + a)^{-\gamma}$ with $a>0$ and $\gamma\in(0,1]$\cite{Zou2006, WangLeng2008}. Tuning $(\lambda_{1},\lambda_{2})$ is selected by spatially blocked cross-validation.

Define $\bm r(\bm\theta_P,\bm\Theta_L)=\bm y-\bm Z\bm\alpha-\bm X\bm\beta_G-\bm X_\delta \bm\Delta$ with $\bm\Delta=[\bm\delta_1^\top,\dots,\bm\delta_p^\top]^\top$. The objective \eqref{eq:obj-main} is convex and separable across groups for the non-smooth term.

\paragraph{Handling the centering constraint.} After each update of $\bm\delta_j$, project it to the degree-weight centered subspace. Let $\{\mathcal C_k\}_{k=1}^K$ be the components and $\bm 1_{\mathcal C_k}$ their indicators. The projector reads
\begin{equation}\label{eq:proj}
    \mathsf{Proj}_D(\bm v) = 
    \bm v - \sum_{k=1}^K \frac{\bm 1_{\mathcal C_k}^\top \bm D \bm v}{\bm 1_{\mathcal C_k}^\top \bm D \bm 1_{\mathcal C_k}} \bm 1_{\mathcal C_k},
\end{equation}
and we enforce $\bm\delta_j\leftarrow \mathsf{Proj}_D(\bm\delta_j)$ every iteration.

\subsection*{An ADMM scheme with closed-form proximal maps}

Introduce $\bm s=\bm r(\bm\theta_P,\bm\Theta_L)$ and $\bm z_j=\bm\delta_j$ for $j=1,\dots,p$, and solve
\begin{equation}\label{eq:admm-form}
    \min_{\bm\theta_P,\bm\Theta_L, \bm s, \{\bm z_j\}}
    \sum_{i=1}^n \rho_\tau(s_i)
    + \lambda_{1}\sum_{j=1}^p w_j \|\bm z_j\|_2
    + \lambda_{2}\sum_{j=1}^p \bm\delta_j^\top \bm L_{\mathrm{sym}} \bm\delta_j
\end{equation}
subject to $\bm s=\bm y-\bm Z\bm\alpha-\bm X\bm\beta_G-\sum_{j=1}^p \bm X_{\odot j}\bm\delta_j$ and $\bm z_j=\bm\delta_j$, together with centering via \eqref{eq:proj}. Using scaled duals $(\bm u,\{\bm v_j\})$ and penalties $(\rho_s,\rho_z)>0$, an ADMM sweep is as follows\cite{Boyd2011}:

\emph{Update of $(\bm\alpha,\bm\beta_G)$.} Solve
\begin{equation*}
    \begin{bmatrix}
        \bm Z^\top\bm Z & \bm Z^\top\bm X \\
        \bm X^\top\bm Z & \bm X^\top\bm X \\
    \end{bmatrix}
    \begin{bmatrix}
        \bm\alpha  \\ 
        \bm\beta_G \\
    \end{bmatrix} 
    =
    \begin{bmatrix}
        \bm Z^\top(\bm y - \sum_j \bm X_{\odot j}\bm\delta_j - \bm s + \bm u) \\
        \bm X^\top(\bm y - \sum_j \bm X_{\odot j}\bm\delta_j - \bm s + \bm u) \\
    \end{bmatrix}.
\end{equation*}
The coefficient matrix depends only on $[\bm Z \ \bm X]$ and can be factorized once.

\emph{Update of $\bm s$.} The proximal map of the check loss is an asymmetric soft-threshold. For each $i$,
\begin{equation}\label{eq:prox-check}
    \bm s \leftarrow \mathrm{prox}_{\rho_s^{-1}\rho_\tau} (\bm y - \bm Z\bm\alpha-\bm X\bm\beta_G-\sum_j \bm X_{\odot j}\bm\delta_j + \bm u),
\end{equation}
with
\begin{equation*}
    \mathrm{prox}_{\gamma \rho_\tau}(v)=
    \begin{cases}
        v - \gamma \tau,      & v > \gamma \tau, \\
        0,                    & -\gamma(1 - \tau) \le v \le \gamma \tau, \\
        v + \gamma(1 - \tau), & v < -\gamma(1 - \tau). \\
    \end{cases}
\end{equation*}

\emph{Update of $\bm\delta_j$.} Holding other blocks fixed, $\bm\delta_j$ solves
\begin{equation}\label{eq:delta-lin}
    \begin{aligned}
        \big(2\lambda_{2}\bm L_{\mathrm{sym}}+\rho_s \bm X_{\odot j}^\top \bm X_{\odot j} + \rho_z \bm I\big) \bm\delta_j
        = 
        \rho_s \bm X_{\odot j}^\top\big(\bm y - \bm Z\bm\alpha - \bm X\bm\beta_G - \sum_{\ell\ne j} \bm X_{\odot \ell}\bm\delta_\ell - \bm s + \bm u\big)
        \\
        + \rho_z (\bm z_j-\bm v_j).
    \end{aligned}
\end{equation}
Solve by sparse Cholesky or conjugate gradients; then project $\bm\delta_j\leftarrow \mathsf{Proj}_D(\bm\delta_j)$.

\emph{Update of $\bm z_j$.} Apply group soft-thresholding,
\begin{equation}\label{eq:group-shrink}
    \bm z_j \leftarrow \Big(1-\frac{\lambda_1 w_j}{\rho_z \|\bm\delta_j+\bm v_j\|_2}\Big)_+ (\bm\delta_j+\bm v_j), \quad j=1, \dots, p.
\end{equation}

\emph{Dual updates and stopping.} Update
\begin{equation*}
    \begin{aligned}
        \bm u &\leftarrow \bm u + \big(\bm y-\bm Z\bm\alpha-\bm X\bm\beta_G-\sum_j \bm X_{\odot j}\bm\delta_j-\bm s\big), \\
        \bm v_j &\leftarrow \bm v_j+(\bm\delta_j-\bm z_j),
    \end{aligned}
\end{equation*}
and stop when primal residuals
\begin{equation*}
    \bm r_s = \bm y - \bm Z\bm\alpha - \bm X\bm\beta_G - \sum_j \bm X_{\odot j}\bm\delta_j - \bm s,
    \quad
    \bm r_{z, j} = \bm\delta_j - \bm z_j
\end{equation*}
and dual residuals
\begin{equation*}
    \bm d_s = \rho_s(\bm s - \bm s^{\mathrm{prev}}), \quad
    \bm d_{z,j} = \rho_z(\bm z_j - \bm z_j^{\mathrm{prev}})
\end{equation*}
satisfy $\|\bm r_s\|_2\le \varepsilon_{\mathrm{pri}}$, $\big(\sum_j \|\bm r_{z,j}\|_2^2\big)^{1/2}\le \varepsilon_{\mathrm{pri}}$, $\|\bm d_s\|_2\le \varepsilon_{\mathrm{dual}}$ and $\big(\sum_j \|\bm d_{z,j}\|_2^2\big)^{1/2}\le \varepsilon_{\mathrm{dual}}$, where tolerances scale with $\sqrt n$ and block magnitudes. Over-relaxation and adaptive penalty updates are often beneficial.

\subsection*{A smoothed proximal–gradient scheme with Moreau smoothing}

Define the smoothed loss $M_h(r)=\min_s\{\rho_\tau(s)+\tfrac{1}{2h}(s-r)^2\}$. Its gradient exists and is $1/h$-Lipschitz:
\[
\nabla_r M_h(r)=\frac{1}{h}\big(r-\mathrm{prox}_{h\rho_\tau}(r)\big).
\]
Let
\begin{equation*}
    \mathcal G_h(\bm\theta_P,\bm\Theta_L)
    =\sum_{i=1}^n M_h \Big(r_i(\bm\theta_P,\bm\Theta_L)\Big)
    +\lambda_2 \sum_{j=1}^p \bm\delta_j^\top \bm L_{\mathrm{sym}} \bm\delta_j,
    \quad
    \mathcal R(\bm\Theta_L)=\lambda_1\sum_{j=1}^p w_j\|\bm\delta_j\|_2.
\end{equation*}
Gradients follow by the chain rule:
\begin{equation*}
    \begin{aligned}
        \nabla_{\bm\alpha}\mathcal G_h &= -\bm Z^\top \bm g, \\
        \nabla_{\bm\beta_G}\mathcal G_h &= -\bm X^\top \bm g, \\
        \nabla_{\bm\delta_j}\mathcal G_h &= -\bm X_{\odot j}^\top \bm g + 2\lambda_2 \bm L_{\mathrm{sym}} \bm\delta_j,
    \end{aligned}
\end{equation*}
with $\bm g=\frac{1}{h}(\bm r-\mathrm{prox}_{h\rho_\tau}(\bm r))$.
Given a step size $t>0$ (via backtracking to ensure sufficient decrease), update
\begin{equation*}
    \bm\alpha \leftarrow \bm\alpha - t \nabla_{\bm\alpha}\mathcal G_h, \quad
    \bm\beta_G \leftarrow \bm\beta_G - t \nabla_{\bm\beta_G}\mathcal G_h,
\end{equation*}
\begin{equation*}
    \bm\delta_j \leftarrow \operatorname{shrink}_{t\lambda_1 w_j} \Big(\bm\delta_j - t \nabla_{\bm\delta_j}\mathcal G_h\Big), \quad
    \operatorname{shrink}_{\kappa}(\bm v) = \Big(1-\frac{\kappa}{\|\bm v\|_2}\Big)_+ \bm v,
\end{equation*}
followed by $\bm\delta_j\leftarrow \mathsf{Proj}_D(\bm\delta_j)$. Nesterov acceleration with standard restart is used. A continuation $h\downarrow h_{\min}$ with $\sqrt n h\to\infty$ preserves asymptotics.

\subsection*{Tuning, initialization, and implementation details}

Tuning $(\lambda_1,\lambda_2)$ and graph hyperparameters (e.g.\ mutual $k$ in $k$-NN or kernel bandwidth) uses \emph{spatially blocked cross-validation}. Partition locations into $K$ spatial folds; for each held-out fold, \emph{construct the graph and Laplacian using only the training folds}, fit the model, and evaluate the held-out check loss. This prevents leakage through cross-fold edges. Adaptive weights use a pilot fit with small $\lambda_1$ and moderate $\lambda_2$, or a purely smooth fit with $\lambda_1=0$; set $w_j=(\|\tilde{\bm\delta}_j\|_2+a)^{-\gamma}$ with a small $a$ for stability. Initialize $(\bm\alpha,\bm\beta_G)$ at the global QR and set $\bm\delta_j=\bm 0$.

Both schemes exploit sparsity. $\bm L_{\mathrm{sym}}$ is stored in a compressed sparse format; matrix–vector products cost $O(nd)$ with average degree $d$. The multipliers $\bm X_{\odot j}$ apply in $O(n)$. In ADMM, systems in \eqref{eq:delta-lin} have the form $\rho_z \bm I + \rho_s \mathrm{diag}(X_{:j}^2) + 2\lambda_2 \bm L_{\mathrm{sym}}$ and admit efficient preconditioned CG with a shared preconditioner (e.g.\ incomplete Cholesky of $\rho_z \bm I + 2\lambda_2 \bm L_{\mathrm{sym}}$). In the smoothed scheme, costs are dominated by proximal evaluations (linear time) and sparse graph multiplications.

\subsection*{Stopping rules and KKT diagnostics}

For ADMM, stopping relies on primal/dual residuals as above, with tolerances
\begin{equation*}
    \varepsilon_{\mathrm{pri}} = \sqrt{n} \varepsilon_{\mathrm{abs}} + \varepsilon_{\mathrm{rel}}\max\left\{\|\bm s\|_2, \left\|\bm y - \bm Z\bm\alpha - \bm X\bm\beta_G - \sum_j \bm X_{\odot j}\bm\delta_j\right\|_2\right\},
\end{equation*}
\begin{equation*}
    \varepsilon_{\mathrm{dual}}=\sqrt{n} \varepsilon_{\mathrm{abs}} + \varepsilon_{\mathrm{rel}}\max\left\{\|\rho_s(\bm s - \bm s^{\mathrm{prev}})\|_2, \left(\sum_j \left\|\rho_z\left(\bm z_j - \bm z_j^{\mathrm{prev}}\right)\right\|_2^2\right)^{1/2} \right\}.
\end{equation*}
For the smoothed scheme, stopping uses relative objective decrease and a KKT residual,
\begin{equation*}
    \big\|\bm Z^\top \bm\psi_\tau(\hat{\bm r})\big\|_2+\big\|\bm X^\top \bm\psi_\tau(\hat{\bm r})\big\|_2
    +\sum_{j=1}^p \mathrm{dist} \Big(\bm X_{\odot j}^\top \bm\psi_\tau(\hat{\bm r})+2\lambda_2 \bm L_{\mathrm{sym}}\hat{\bm\delta}_j,\ \lambda_1 w_j \partial\|\hat{\bm\delta}_j\|_2\Big),
\end{equation*}
where $\bm\psi_\tau(\bm r)$ stacks $\psi_\tau(r_i)=\tau-\mathbf 1\{r_i<0\}$. The distance to the group subdifferential is
\[
\mathrm{dist}( \bm v,\ \lambda_1 w_j \partial\|\hat{\bm\delta}_j\|_2)=
\begin{cases}
\left\|\bm v-\lambda_1 w_j \dfrac{\hat{\bm\delta}_j}{\|\hat{\bm\delta}_j\|_2}\right\|_2,& \hat{\bm\delta}_j\neq 0,\\[6pt]
\max\{\|\bm v\|_2-\lambda_1 w_j,\ 0\},& \hat{\bm\delta}_j=0.
\end{cases}
\]

\subsection*{Practical choices}

Penalty scales can be anchored by heuristics. A convenient grid sets $\lambda_2$ around the median of the nonzero eigenvalues of $\bm L_{\mathrm{sym}}$ divided by $n$, and $\lambda_1$ around the empirical $(1-\tau)\tau$ score scale times $\sqrt{\log p/n}$. In noisy data, a continuation schedule that starts with larger $\lambda_2$ and decreases gradually stabilizes both solvers. When repeated locations occur, connect coincident points by a high-weight edge or merge them.

\section{Asymptotic Theory}\label{sec:theory}

We analyze the SS--SVCQR estimator in Section~\ref{sec:model}. The parametric block is $\bm\theta_P=(\bm\alpha^\top,\bm\beta_G^\top)^\top\in\mathbb R^{q+p}$, and the local deviations are $\bm\Theta_L=\{\bm\delta_j\in\mathbb R^n: j=1,\dots,p\}$, each degree-weight centered on every graph component.

\paragraph{Estimator and penalty.}
At quantile level $\tau$,
\begin{equation*}
    \begin{aligned}
        \mathcal L_n(\bm\theta_P,\bm\Theta_L)
        =
        \sum_{i=1}^n \rho_\tau\Big(Y_i - \bm Z_i^\top\bm\alpha - \bm X_i^\top\bm\beta_G - \sum_{j=1}^p X_{ij} \delta_j(u_i)\Big)
        +
        \sum_{j=1}^p \lambda_{1n} w_{j,n} \|\bm\delta_j\|_2 \\
        +
        \sum_{j=1}^p \lambda_{2n} \bm\delta_j^\top \bm L_{\mathrm{sym}}\bm\delta_j,
    \end{aligned}
\end{equation*}
subject to $\bm 1_{\mathcal C}^\top \bm D\bm\delta_j=0$ for every component $\mathcal C$ and every $j$. Weights $w_{j,n}>0$ can be adaptive, e.g.\ $w_{j,n}=(\|\tilde{\bm\delta}_j\|_2 + a_n)^{-\gamma}$ with $a_n\downarrow 0$, $\gamma\in(0,1]$.

\paragraph{Population target and notation.}
Let $(\bm\alpha_0,\bm\beta_0)$ denote the true parameters in the sense of \eqref{eq:qexog}\cite{FanLi2001}. Write $\beta_{0,j}(u)=\beta_{G0,j}+\delta_{0,j}(u)$ with degree-weight centered $\delta_{0,j}(\cdot)$ and define
\[
    \bm\delta_{j0}=\big(\delta_{0,j}(u_1),\dots,\delta_{0,j}(u_n)\big)^\top, 
    \,
    \mathcal A_0=\{j: \bm\delta_{j0}=\bm 0\},
    \, 
    \mathcal A_0^c=\{1,\dots,p\}\setminus\mathcal A_0.
\]

\begin{assumption}[Sampling design and graph]\label{ass:A1}
Locations $u_i$ are i.i.d.\ with a density bounded away from zero and infinity on a compact $\mathcal U\subset\mathbb R^2$. The graph is constructed either (i) by a symmetric mutual $k_n$-nearest-neighbor rule with $k_n$ bounded away from zero and $k_n=O(1)$ (node-level estimation focus), or (ii) by a symmetric compact-support kernel with bandwidth bounded away from zero. With probability tending to one the graph has $O(1)$ connected components. Degrees satisfy $\max_i d_i=O_p(1)$ and $\min_i d_i\ge c>0$ with high probability.
\end{assumption}
Intuitively, Assumption 1 requires that locations be well spread over a compact region and that the $k$–NN graph does not degenerate as $n$ grows.

\begin{remark}
If one targets continuum limits of $\delta_j(\cdot)$, a growing-degree regime can be imposed (e.g., $k_n\to\infty$, $k_n/n\to 0$ or shrinking bandwidth), together with a suitable scaling of the Laplacian; our theory focuses on node-level estimation and keeps $d_i=O_p(1)$.
\end{remark}

\begin{assumption}[Covariates and Gram matrices]\label{ass:A2}
Covariates are uniformly bounded almost surely, $\|\bm Z_i\|_\infty\le C_Z$ and $\|\bm X_i\|_\infty\le C_X$. The Gram matrices $n^{-1}\bm Z^\top\bm Z$ and $n^{-1}\bm X^\top\bm X$ converge to positive definite limits. For each $j$, $\bm X_{\odot j}$ is not identically zero on any component.
\end{assumption}

\begin{assumption}[Errors and local density]\label{ass:A3}
Let $\varepsilon_i=Y_i-Q_\tau(Y_i\mid \bm Z_i,\bm X_i,u_i)$. Then $\Pr(\varepsilon_i\le 0\mid \bm Z_i,\bm X_i,u_i)=\tau$ a.s., the conditional density $f_{\varepsilon\mid Z,X,U}(t)$ exists and is continuous near $t=0$, and $0<c_f\le f_{\varepsilon\mid Z,X,U}(0)\le C_f<\infty$ a.s.
\end{assumption}

\begin{assumption}[Smoothness of local deviations]\label{ass:A4}
For each $j\in\mathcal A_0^c$, the true deviation vector satisfies $\bm\delta_{j0}^\top \bm L_{\mathrm{sym}}\bm\delta_{j0}\le C_j$ with constants $C_j<\infty$ not depending on $n$.
\end{assumption}

\begin{assumption}[Penalty sequences and adaptive weights]\label{ass:A5}
The smoothness penalty obeys $\lambda_{2n}\to 0$ and $n\lambda_{2n}\to\infty$. The sparsity penalty satisfies $\lambda_{1n}\to 0$ and $\sqrt n \lambda_{1n}\to\infty$. Adaptive weights obey $\sup_{j\in\mathcal A_0^c}w_{j,n}=O_p(1)$ and $\inf_{j\in\mathcal A_0}w_{j,n}\ge c_w>0$ with probability tending to one.
\end{assumption}

\begin{assumption}[Beta–min separation]\label{ass:A6}
There exists $c_\beta>0$ such that for all $j\in\mathcal A_0^c$,
$\ \|\bm\delta_{j0}\|_2 \ge c_\beta \lambda_{1n} w_{j,n}$
for all sufficiently large $n$ with probability tending to one.
\end{assumption}

\begin{assumption}[Restricted curvature for the quantile loss]\label{ass:A7}
Let $\bm G_i=(\bm Z_i^\top,\bm X_i^\top)^\top$ and define
\[
    \bm M = \mathrm E\big[f_{\varepsilon\mid Z,X,U}(0) \bm G_i\bm G_i^\top\big],
    \quad
    \bm V = \tau(1-\tau) \mathrm E\big[\bm G_i\bm G_i^\top\big].
\]
Then $\bm M$ is positive definite. Moreover, for any $\bm v$ in a neighborhood of the origin,
\[
    n^{-1}\sum_{i=1}^n\big\{\rho_\tau(\varepsilon_i - \bm G_i^\top \bm v) - \rho_\tau(\varepsilon_i)\big\}
    =
    \frac12 \bm v^\top \bm M\bm v + o_p(\|\bm v\|_2^2),
\]
uniformly over events where $\|\bm v\|_2=O(n^{-1/2})$.
\end{assumption}

\begin{remark}
Our asymptotics cover fixed $p$. For growing $p$, additional group-sparsity conditions (e.g., restricted eigenvalues) are required and are left for future work.
\end{remark}

\paragraph{KKT characterization.}
Let $\psi_\tau(r)=\tau-\mathbf 1\{r<0\}$ and $\hat r_i$ be the fitted residuals. At any global minimizer $(\hat{\bm\theta}_P,\hat{\bm\Theta}_L)$, the parametric block satisfies
\[
    \sum_{i=1}^n \psi_\tau(\hat r_i) \bm G_i = \bm 0.
\]
For each $\bm\delta_j$,
\[
    \bm X_{\odot j}^\top \bm\psi_\tau(\hat{\bm r}) + 2 \lambda_{2n} \bm L_{\mathrm{sym}}\hat{\bm\delta}_j + \bm D\bm\mu_{j} \in -\lambda_{1n} w_{j, n} \partial\|\hat{\bm\delta}_j\|_2,
\]
with centering constraints $\bm 1_{\mathcal C}^\top \bm D\hat{\bm\delta}_j=0$ on every component and multipliers $\bm\mu_j$.

\subsection*{Deviation rates and selection consistency}

\begin{theorem}[Deviation rates and selection consistency]\label{thm:delta-selection}
Suppose Assumptions~\ref{ass:A1}–\ref{ass:A6} hold. Then:

\begin{enumerate}
\item For each $j\in\mathcal A_0^c$,
\[
    \frac{1}{n}\sum_{i=1}^n\big\{\hat\delta_j(u_i)-\delta_{0,j}(u_i)\big\}^2
    = O_p\Big(\frac{1}{n \lambda_{2n}} + \lambda_{2n} \mathcal S_j\Big),
\]
where $\mathcal S_j$ depends on $\bm\delta_{j0}^\top \bm L_{\mathrm{sym}}\bm\delta_{j0}$ and the spectrum of $\bm L_{\mathrm{sym}}$. The bound is balanced at $\lambda_{2n}\asymp (n \mathcal S_j)^{-1/2}$.

\item The adaptive group penalty correctly separates globally constant and truly varying covariates:
\[
    \Pr\big(\hat{\bm\delta}_j = \bm 0\ \text{for all } j\in\mathcal A_0\big) \rightarrow 1,
    \quad
    \Pr\big(\hat{\bm\delta}_j\ne\bm 0\ \text{for all } j\in\mathcal A_0^c\big)\rightarrow 1.
\]
On the event of correct selection, the deviations satisfy the rate in part (i).
\end{enumerate}
\end{theorem}

\subsection*{Parametric limit distribution and oracle property}

\begin{theorem}[Asymptotic normality and oracle property]\label{thm:normal-oracle}
Suppose Assumptions~\ref{ass:A1}–\ref{ass:A7} hold. Then
\[
    \sqrt n (\hat{\bm\theta}_P - \bm\theta_{P0})
    \overset{d}{\longrightarrow}\mathcal N\big(\bm 0, \bm M^{-1}\bm V\bm M^{-1}\big).
\]
Moreover, on $\{\hat{\mathcal A} = \mathcal A_0\}$ (which holds with probability tending to one by Theorem~\ref{thm:delta-selection}), the joint limit of $(\hat{\bm\theta}_P, \{\hat{\bm\delta}_j\}_{j\in\mathcal A_0^c})$ coincides with that of the oracle estimator that knows $\bm\delta_{j0}=\bm 0$ for all $j\in\mathcal A_0$.
\end{theorem}

\begin{lemma}[Cross-term control]\label{lem:crossterm}
Under bounded covariates and fixed $p$, if $\lambda_{2n}\asymp n^{-1/2}$, then
\[
\left\|\frac{1}{\sqrt n}\sum_{i=1}^n f_{\varepsilon\mid Z,X,U}(0)\Big(\sum_{j=1}^p X_{ij} (\hat\delta_j(u_i)-\delta_{0,j}(u_i))\Big)\bm G_i\right\|_2
= o_p(1).
\]
\end{lemma}
\paragraph{Standard error estimation.} Estimate the asymptotic variance using the plug-in
\begin{equation*}
    \hat{\bm M} = n^{-1}\sum_{i=1}^n \hat f_i(0) \bm G_i\bm G_i^\top, 
    \quad
    \hat{\bm V} = \tau(1 - \tau) n^{-1}\sum_{i=1}^n \bm G_i\bm G_i^\top,
\end{equation*}
where $\hat f_i(0)$ is a local estimator of $f_{\varepsilon\mid Z,X,U}(0)$ (e.g., kernel density on residuals) or via a smoothed objective with bandwidth $h_n\downarrow 0$ and $\sqrt n h_n\to\infty$. When spatial dependence is a concern, one may use spatial block bootstrap or Conley-type HAC corrections in the sandwich step.

\section{Simulation Study}\label{sec:simulation}

We conducted two complementary simulation experiments. \emph{Scenario~1} uses a single relatively dense data set to provide a visual check that the proposed estimator can correctly recover the spatial structure of the coefficients under a variety of error distributions.  \emph{Scenario~2} is a Monte Carlo study that summarizes selection accuracy, estimation error, and predictive performance over repeated samples.

Both scenarios share the same basic design for spatial locations, covariates, and true coefficient surfaces; only the sample size, error distribution, and the way the results are summarized differ.  The estimation algorithm (ADMM implementation of SS--SVCQR with two-stage adaptive group penalty and spatially blocked cross-validation) is exactly as described in Section~\ref{sec:estimation} and is therefore not repeated here.

\paragraph{Settings} We place $n$ monitoring sites on the unit square $\mathcal U=[0,1]^2$. The locations $u_i=(u_{i1},u_{i2})^\top$ are i.i.d.\ $\mathrm{Unif}([0,1]^2)$, mimicking an irregular but dense network of stations. At each site, we observe a global block $\bm Z_i\in\mathbb R^3$ and a potentially spatially varying block $\bm X_i\in\mathbb R^4$,
\begin{equation*}
    Z_{i0}\equiv 1, \quad
    Z_{i1}\sim\mathcal N(0,1), \quad
    Z_{i2}\sim\mathcal N(0,1), \quad
    X_{ij}\sim\mathcal N(0,1), j = 1, \dots, 4,
\end{equation*}
independently in $i$ and $j$.  

The conditional $\tau$th quantile is
\begin{equation*}
    Q_\tau(Y_i\mid \bm Z_i,\bm X_i,u_i)
    =
    \bm Z_i^\top\bm\alpha_0 + \sum_{j=1}^4 X_{ij} \beta_{0,j}(u_i),
    \quad
    \beta_{0,j}(u)=\beta_{G0,j} + \delta_{0,j}(u),
\end{equation*}
with global effects
\begin{equation*}
    \bm\alpha_0=(3,-1.0,1.5)^\top, 
    \quad
    \bm\beta_{G0}=(5,0,2.5,0)^\top,
\end{equation*}
so that $X_1$ and $X_3$ have nonzero global components, while $X_2$ and $X_4$ would be purely global if their spatial deviations vanished.

We design the deviation fields so that exactly two covariates are truly local:
\begin{align*}
    \delta_{0, 1}(u)
    &= A_1\left\{\sin(2\pi u_1)\cos(2\pi u_2) + c_1(u_1 - 0.5)\right\}, \\
    \delta_{0, 2}(u) & \equiv 0, \\
    \delta_{0, 3}(u)
    &= A_3\left\{(u_1 - 0.5)^2+(u_2 - 0.5)^2\right\}, \\
    \delta_{0, 4}(u) & \equiv 0,
\end{align*}
where $(A_1,c_1,A_3)$ are fixed constants chosen so that the signal amplitude is comparable to the random noise.  
In the implementation we evaluate these fields at the observed sites, construct a $k$-nearest-neighbour graph on $\{u_i\}$, and apply the same degree-weighted centering operator as in Section~\ref{sec:model} so that each $\delta_{0,j}(\cdot)$ has zero degree-weighted mean on every connected component.  Under this centering, $\beta_{G0,j}$ is the degree-weighted spatial average of $\beta_{0,j}(\cdot)$ and $\delta_{0,j}(\cdot)$ represents pure spatial departures.

\subsection{Recovery of Spatial Structure}\label{sec:sim-scenario1}

Scenario~1 uses a single, relatively large dataset ($n=1000$) for each error distribution and focuses on qualitative recovery of the spatial deviation fields.  Responses are generated as
\begin{equation*}
    Y_i = Q_\tau(Y_i\mid \bm Z_i, \bm X_i, u_i) + \varepsilon_i,
\end{equation*}
with $\varepsilon_i$ independent of $(\bm Z_i, \bm X_i, u_i)$.

We consider several error distributions to probe robustness:
\begin{enumerate}[(a)]
  \item Gaussian: $\varepsilon_i\sim\mathcal N(0,\sigma^2)$;
  \item asymmetric Laplace (ALD) with parameter $\tau=0.5$, which aligns the conditional median exactly with $Q_\tau(\cdot)$;
  \item Student-$t_3$ noise;
  \item contaminated normal: a $90\%$--$10\%$ mixture of $\mathcal N(0,\sigma^2)$ and $\mathcal N(0,25\sigma^2)$;
  \item Cauchy errors;
  \item heteroskedastic $t_3$ errors,
  \[
    \varepsilon_i = \sigma(u_i) \eta_i,\qquad
    \sigma(u_i)=0.5+0.5u_{i1},\ \eta_i\sim t_3,
  \]
  which induce a left--right gradient in noise scale.
\end{enumerate}

For the contaminated normal and Cauchy cases we additionally flag as \textit{outliers} those observations with $|\varepsilon_i|$ exceeding a high cut-off (e.g.\ $2.5$ times the Gaussian standard deviation).  These outliers are indicated as hollow circles overlaid on the \emph{true} fields in the plots, illustrating that the proposed quantile-based estimator is robust to a small number of extreme data points.

For each error scenario, we fit SS--SVCQR to $\tau=0.5$ using the spatially blocked cross-validation scheme of Section~\ref{sec:estimation} to select $(\lambda_1,\lambda_2)$, and then compare the estimated deviations $\hat\delta_j(u_i)$ with the truth.  

Figures~\ref{fig:normal}--\ref{fig:hetero} show, for each covariate $X_j$, three panels: the true deviation $\delta_{0,j}(u_i)$, the estimated deviation $\hat\delta_j(u_i)$, and their difference (error).  All panels share a common color scale so that amplitude and sign can be compared directly across variables and scenarios.  In all cases $X_2$ and $X_4$ appear essentially flat, confirming that the method does not spuriously create spatial variation for truly global covariates, while the complex patterns for $X_1$ and $X_3$ are well recovered even under heavy-tailed and heteroskedastic noise.

\begin{figure}[ht]
    \centering
    \includegraphics[width=1\linewidth]{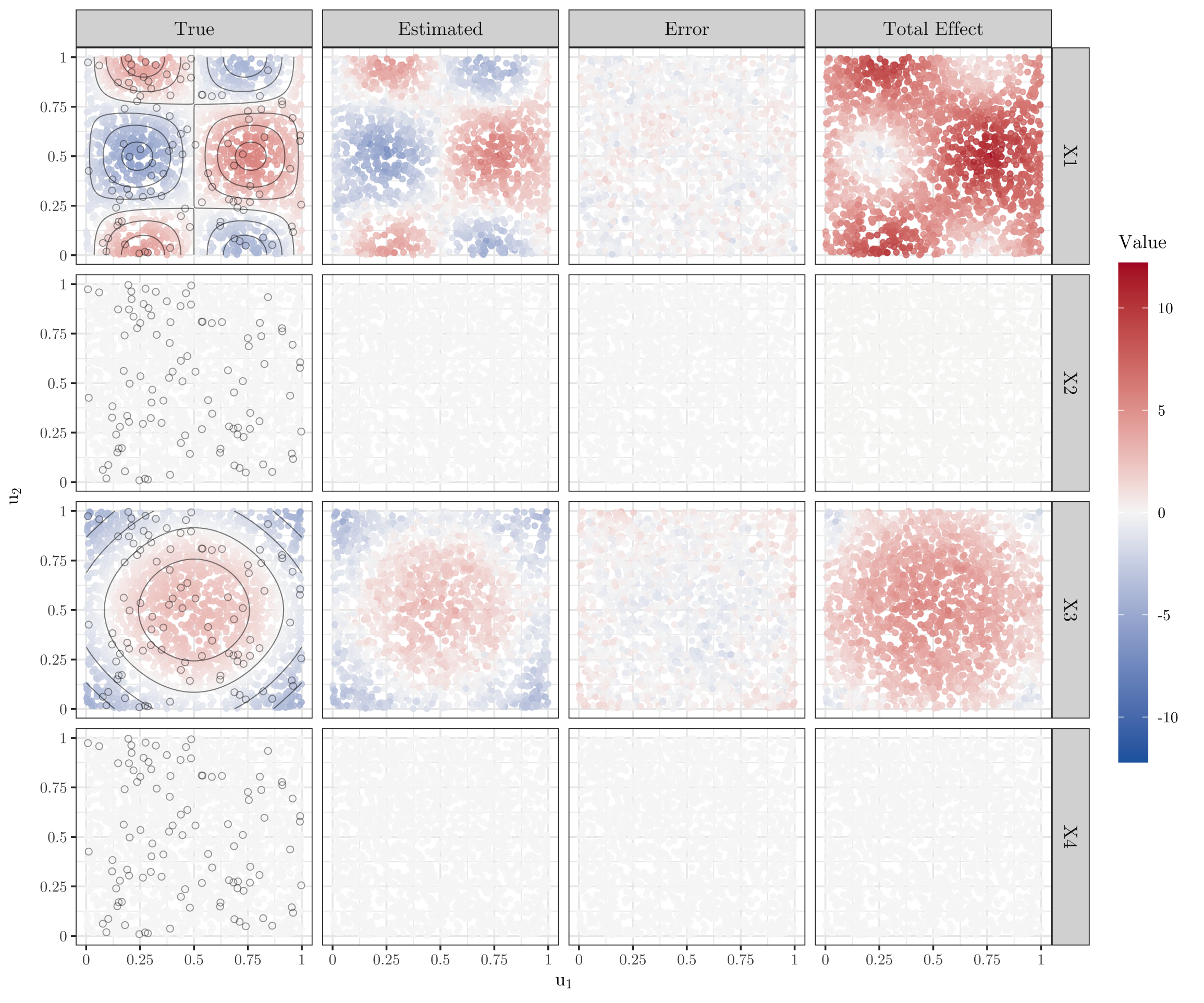}
    \caption{Scenario~1 under normal errors.}
    \label{fig:normal}
\end{figure}

\begin{figure}[ht]
    \centering
    \includegraphics[width=1\linewidth]{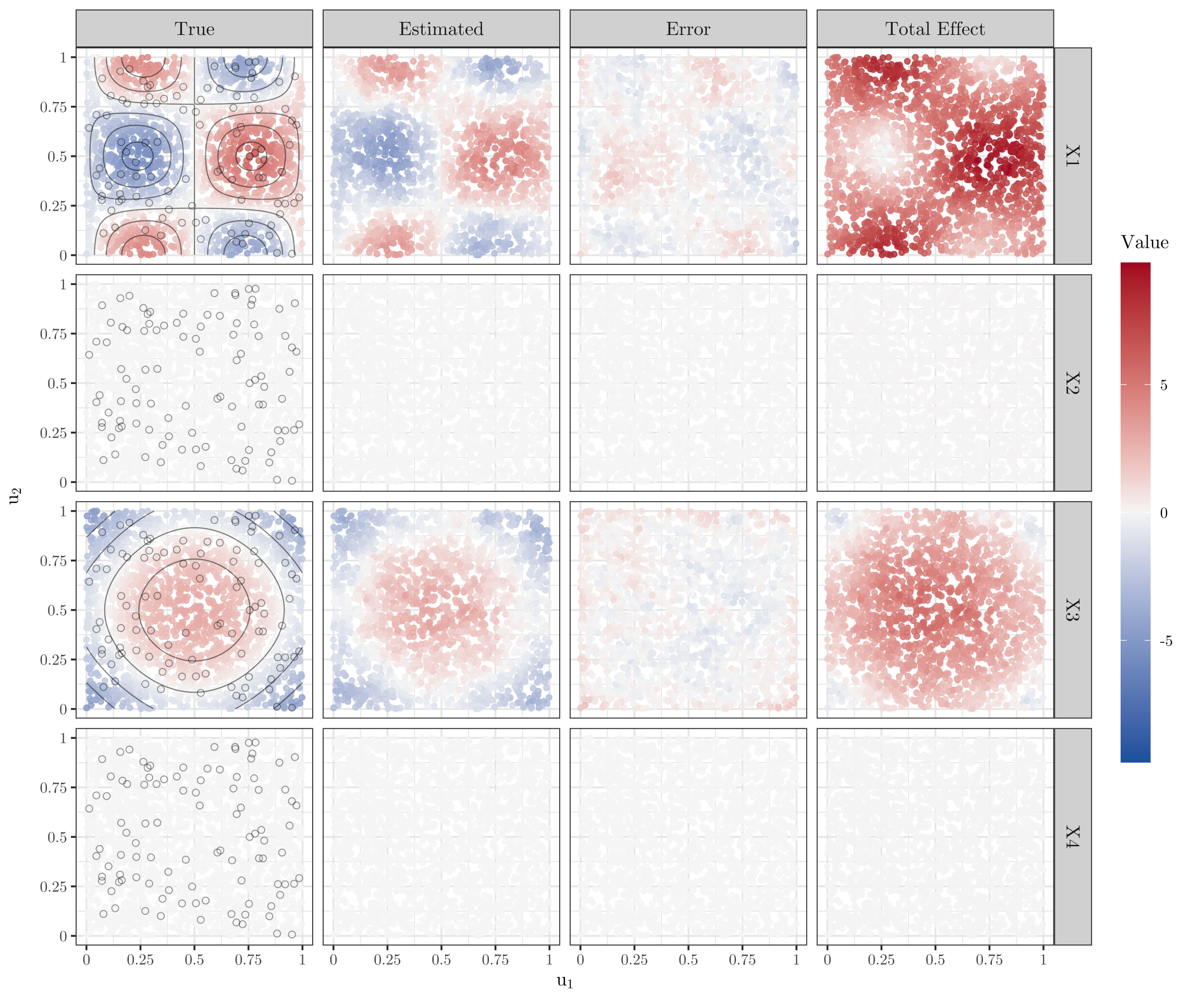}
    \caption{Scenario~1 under asymmetric Laplace errors.}
    \label{fig:ald}
\end{figure}

\begin{figure}[ht]
    \centering
    \includegraphics[width=1\linewidth]{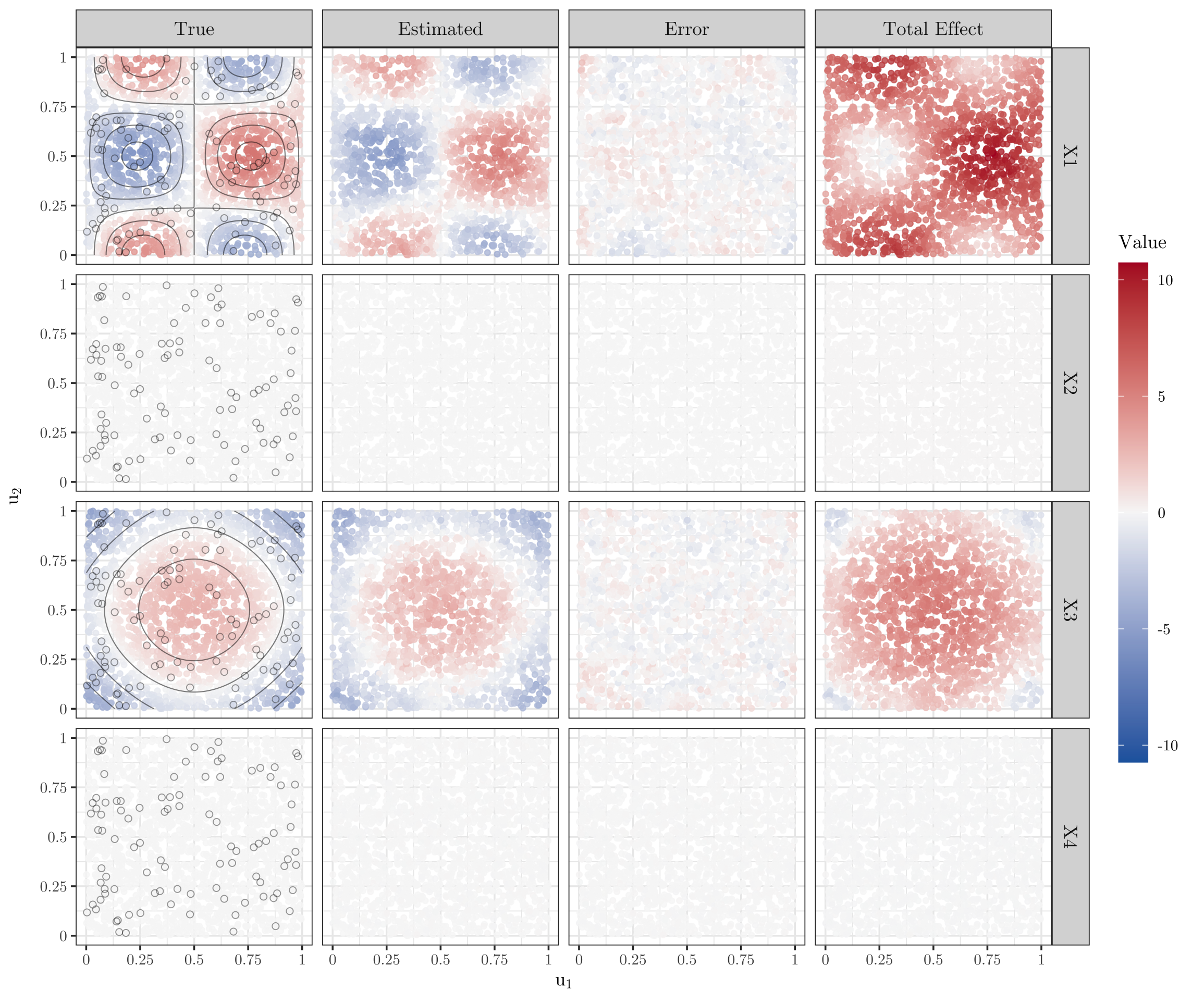}
    \caption{Scenario~1 under Student-$t_3$ errors.}
    \label{fig:t3}
\end{figure}

\begin{figure}[ht]
    \centering
    \includegraphics[width=1\linewidth]{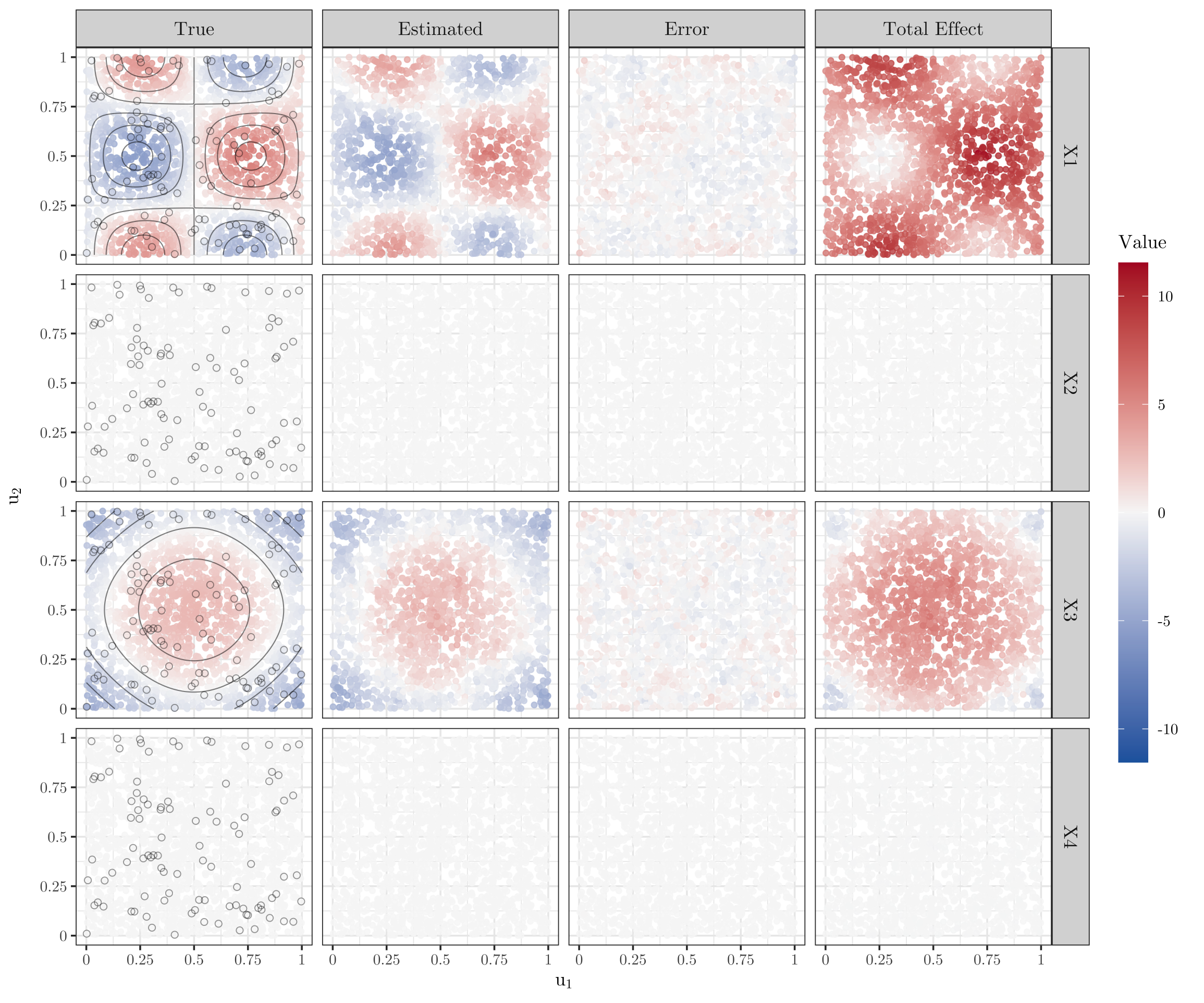}
    \caption{Scenario~1 under contaminated normal errors (hollow circles mark gross outliers).}
    \label{fig:contam_norm}
\end{figure}

\begin{figure}[ht]
    \centering
    \includegraphics[width=1\linewidth]{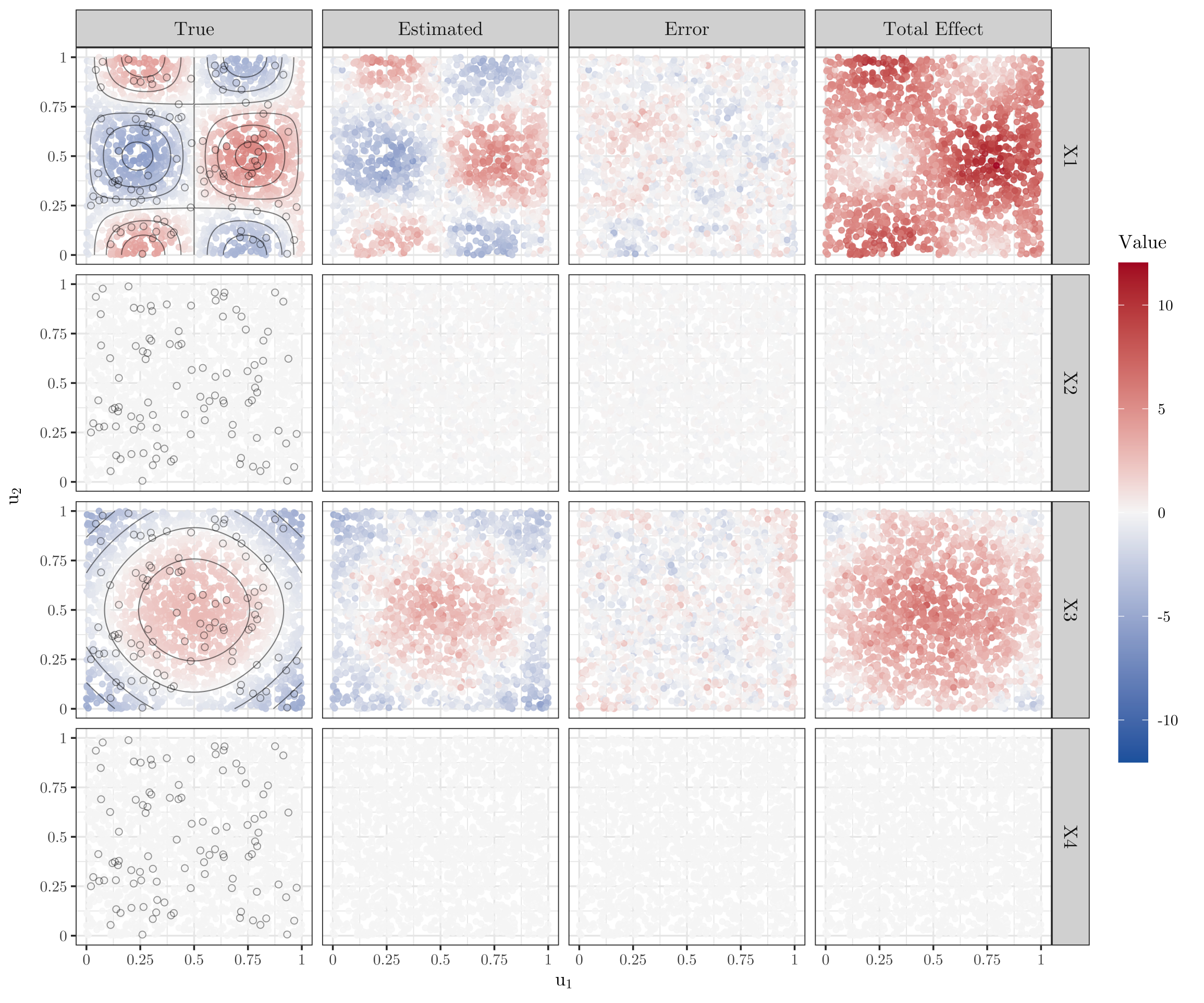}
    \caption{Scenario~1 under Cauchy errors (hollow circles mark extreme outliers).}
    \label{fig:cauchy}
\end{figure}

\begin{figure}[ht]
    \centering
    \includegraphics[width=1\linewidth]{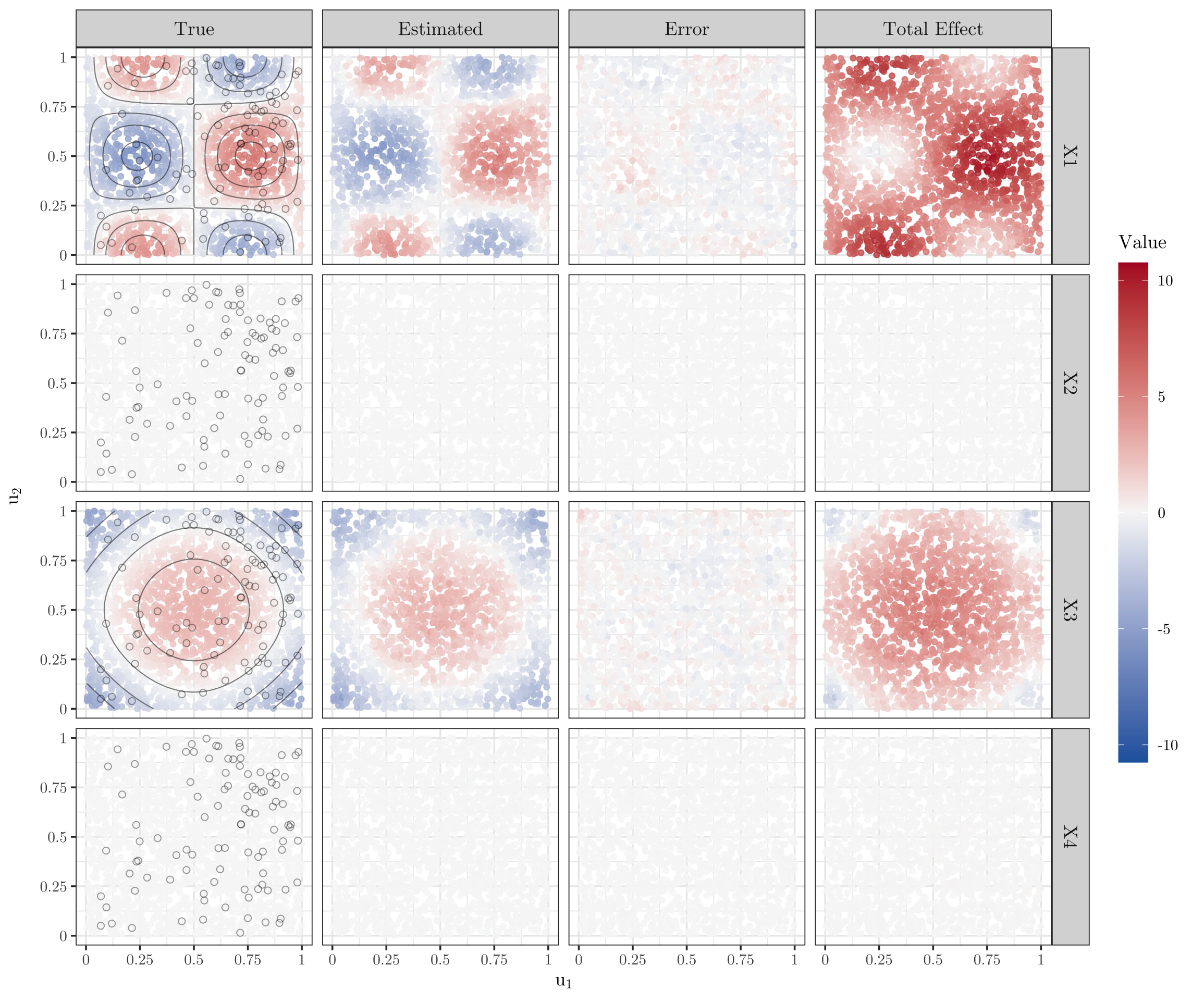}
    \caption{Scenario~1 under heteroskedastic $t_3$ errors with spatially varying scale.}
    \label{fig:hetero}
\end{figure}

Figures~\ref{fig:normal}–\ref{fig:hetero} summarize the results in the six error scenarios. In each figure, the rows correspond to covariates $X_1, \dots, X_4$, while the four columns, from left to right, show the true spatial deviation, the estimated deviation from SS--SVCQR, the pointwise error and the total effect $\beta_{G,j} + \delta_j(u)$ at the target quantile.

The first column shows the simulated \textit{ground truth}: the deviation field $\delta_{0,j}(u)$ evaluated at the sampled locations, together with smoothed contour lines that make the underlying spatial pattern easier to discern. For $X_1$ this reveals an oscillatory wave-like pattern superimposed on a mild trend, while $X_3$ exhibits a smooth dome-shaped hotspot; by contrast, $X_2$ and $X_4$ are truly global and their deviation fields are essentially flat, apart from the overlaid outlier points in the heavy-tailed scenarios.

The second column contains the corresponding estimates $\hat\delta_j(u)$ and shows that the proposed estimator reproduces the main qualitative features of the true fields: the sign changes, level sets and apparent boundaries for $X_1$ and $X_3$ closely track those in the first column, while the panels for $X_2$ and $X_4$ remain visually featureless, indicating that spurious spatial variation is not introduced for covariates that are in fact global.

The third column displays the difference $\hat\delta_j(u_i)-\delta_{0,j}(u_i)$; here the colors fluctuate around zero without forming coherent spatial patches, and there is no obvious residual spatial correlation. This noise-like appearance of the error field indicates that the dominant spatial structure has been absorbed by the model and that the remaining lack-of-fit behaves as idiosyncratic noise rather than as a missing large-scale trend.

Finally, the fourth column reports the total effect of each covariate in location $u_i$, combining the global component and the spatial deviation. For $X_1$ and $X_3$ the panels are dominated by warm colors, suggesting an overall positive association if only a single global coefficient were to be examined; however, the first column reveals sizeable regions where the local deviation is negative and partially cancels the global effect. Without explicitly separating a global level from a spatial deviation, a purely global quantile regression would report a single positive coefficient and completely miss these local reversals, potentially leading to misleading substantive conclusions.

Across all six error mechanisms, including heavy-tailed, contaminated, and heteroskedastic designs, the same qualitative picture emerges: the estimated deviation fields in the second column closely match the truth, the error fields in the third column are spatially structureless, and the total-effect panels in the fourth column make clear that spatial nonstationarity can be substantial even when the overall effect is positive on average.

\subsection{Monte Carlo Assessment}\label{sec:sim-scenario2}

Scenario~2 quantifies performance over repeated samples.  The spatial design, covariates and coefficient surfaces are exactly as in the common setup above, but we now generate $R=100$ independent datasets with sample size $n=1000$ for each error scenario of interest (e.g. Gaussian, ALD, Student-$t_3$, and Cauchy).  For each replicate we fit SS--SVCQR at $\tau=0.5$ with cross-validated $(\lambda_1,\lambda_2)$ and record:
\begin{itemize}
  \item \emph{Global parameter error}
  \[
  \mathrm{PE}_\theta = 
  \|\hat{\bm\alpha}-\bm\alpha_0\|_2
  + \|\hat{\bm\beta}_G-\bm\beta_{G0}\|_2.
  \]
  \item \emph{Deviation mean-squared error} for each covariate,
  \[
  \mathrm{MSE}_{\delta,j}
  = \frac{1}{n}\sum_{i=1}^n
    \big\{\hat\delta_j(u_i)-\delta_{0,j}(u_i)\big\}^2,
    \qquad j=1,\dots,4.
  \]
  \item \emph{Local/global classification accuracy.}  
  We declare $X_j$ to be ``local'' if
  \[
  \sqrt{\frac{1}{n}\sum_{i=1}^n \hat\delta_j(u_i)^2} > \kappa
  \]
  for a small threshold $\kappa$, and ``global'' otherwise.  
  Using the true set of local covariates $\{X_1,X_3\}$, we compute the number of true positives (TP), true negatives (TN), false positives (FP) and false negatives (FN), together with sensitivity $\mathrm{TP}/(\mathrm{TP}+\mathrm{FN})$ and specificity $\mathrm{TN}/(\mathrm{TN}+\mathrm{FP})$.
  \item \emph{Predictive performance} measured by the average check loss on an independent test set of size $n_{\mathrm{test}}$,
  \[
  \mathrm{CL}_{\mathrm{test}}
  =
  \frac{1}{n_{\mathrm{test}}}\sum_{i=1}^{n_{\mathrm{test}}}
  \rho_\tau \Big(
    Y_i^{\mathrm{test}} -
    \hat Q_\tau(Y_i^{\mathrm{test}}
    \mid \bm Z_i^{\mathrm{test}},\bm X_i^{\mathrm{test}},u_i^{\mathrm{test}})
  \Big).
  \]
\end{itemize}

In order to objectively evaluate the performance of the proposed SS--SVCQR model, we compare it with five representative competitors. 
(i) Global quantile regression (QR) uses a purely global specification and thus ignores spatial nonstationarity\cite{quantregPackage}. 
(ii) Standard geographically weighted regression (GWR) treats all covariates as spatially varying\cite{GWmodelPackage}. 
(iii) Mixed GWR assumes that the local and global covariates are known a priori and fit a partially varying model; this can be regarded as an oracle benchmark for methods that must learn the local/global structure. 
(iv) Quantile generalized additive models (QGAM) allow flexible nonlinear effects but do not directly target spatially varying coefficients\cite{Fasiolo2020}. 
(v) The classical spatially varying coefficient (SVC) model is fitted by least squares without quantile weighting\cite{Hastie1993, Gelfand2003SVC}. 
These methods span the spectrum from purely global to fully spatially varying specifications and from mean to quantile regression, and they are chosen to assess three aspects: (a) accuracy of global and varying coefficient estimation, (b) ability to recover the true local versus global covariates, and (c) predictive performance under a range of error distributions. We report the detailed results of the Monte Carlo simulation in the six tables below.\footnote{In Tables \ref{tab:model_QR}--\ref{tab:model_SVCQR}, the values in parentheses represent the standard deviation corresponding to that metric.}
\begin{landscape}
\begin{table}[H]
    \centering
    \caption{Monte Carlo result for quantile regression (QR)}
    \label{tab:model_QR}
    \setlength{\tabcolsep}{4pt}
    \renewcommand{\arraystretch}{1.1}
    \small
    \begin{tabular}{lccccccccc}
        \toprule
               &               &                    \multicolumn{4}{c}{MSE}                    &       &       &               \\
        \cmidrule(lr){3-6}
        Error  & PE$_{\theta}$ & $X_1$         & $X_2$         & $X_3$         & $X_4$         & Sens. & Spec. & CL            \\
        \midrule
        normal & 0.571 (0.157) & 6.824 (0.359) & 0.021 (0.026) & 3.081 (0.166) & 0.024 (0.036) & 0 (0) & 1 (0) & 3.040 (0.337) \\
        ALD    & 0.740 (0.215) & 6.837 (0.350) & 0.039 (0.060) & 3.099 (0.188) & 0.043 (0.062) & 0 (0) & 1 (0) & 3.039 (0.364) \\
        hetero & 0.431 (0.117) & 1.755 (0.244) & 0.080 (0.027) & 0.631 (0.117) & 0.077 (0.028) & 1 (0) & 0 (0) & 2.847 (0.406) \\
        t3     & 0.587 (0.162) & 6.814 (0.363) & 0.019 (0.030) & 3.089 (0.183) & 0.024 (0.032) & 0 (0) & 1 (0) & 2.998 (0.370) \\
        contam & 0.590 (0.188) & 6.824 (0.368) & 0.020 (0.034) & 3.085 (0.177) & 0.023 (0.026) & 0 (0) & 1 (0) & 3.005 (0.317) \\
        cauchy & 0.701 (0.215) & 6.839 (0.361) & 0.033 (0.038) & 3.100 (0.180) & 0.032 (0.049) & 0 (0) & 1 (0) & 5.803 (38.59) \\
        \bottomrule
    \end{tabular}
\end{table}

\begin{table}[H]
    \centering
    \caption{Monte Carlo result for geographically weighted regression}
    \label{tab:model_GWR}
    \setlength{\tabcolsep}{4pt}
    \renewcommand{\arraystretch}{1.1}
    \small
    \begin{tabular}{lccccccccc}
        \toprule
               &               &                    \multicolumn{4}{c}{MSE}                    &       &       &               \\
        \cmidrule(lr){3-6}
        Error  & PE$_{\theta}$ & $X_1$         & $X_2$         & $X_3$         & $X_4$         & Sens. & Spec. & CL            \\
        \midrule
        normal & 0.471 (0.108) & 2.155 (0.207) & 0.064 (0.031) & 0.705 (0.135) & 0.065 (0.028) & 1 (0) & 0 (0) & 2.844 (0.342) \\
        ALD    & 0.615 (0.186) & 2.233 (0.244) & 0.152 (0.071) & 0.799 (0.196) & 0.134 (0.065) & 1 (0) & 0 (0) & 2.956 (0.401) \\
        hetero & 0.483 (0.115) & 2.173 (0.206) & 0.071 (0.040) & 0.728 (0.129) & 0.070 (0.039) & 1 (0) & 0 (0) & 2.809 (0.405) \\
        t3     & 0.508 (0.122) & 2.189 (0.216) & 0.083 (0.046) & 0.743 (0.135) & 0.083 (0.044) & 1 (0) & 0 (0) & 2.825 (0.403) \\
        contam & 0.516 (0.129) & 2.179 (0.230) & 0.090 (0.042) & 0.730 (0.144) & 0.090 (0.037) & 1 (0) & 0 (0) & 2.877 (0.326) \\
        cauchy & 124.1 (1099.3)
               & \num{9.3e5} (\num{9.3e6})
               & \num{1.2e5} (\num{1.2e6})
               & \num{8.9e5} (\num{8.9e6})
               & \num{1.5e5} (\num{1.5e6})
               & 1 (0)
               & 0 (0)
               & 40.78 (305.3) \\
        \bottomrule
    \end{tabular}
\end{table}
\end{landscape}

\begin{landscape}
\begin{table}[H]
    \centering
    \caption{Monte Carlo result for mixed geographically weighted regression}
    \label{tab:model_MGWR}
    \setlength{\tabcolsep}{4pt}
    \renewcommand{\arraystretch}{1.1}
    \small
    \begin{tabular}{lccccccccc}
        \toprule
               &               &                    \multicolumn{4}{c}{MSE}                    &       &       &               \\
        \cmidrule(lr){3-6}
        Error  & PE$_{\theta}$ & $X_1$         & $X_2$         & $X_3$         & $X_4$         & Sens. & Spec. & CL            \\
        \midrule
        normal & 0.255 (0.063) & 0.568 (0.117) & 0.094 (0.022) & 0.238 (0.053) & 0.096 (0.027) & 1 (0) & 0 (0) & 2.869 (0.384) \\
        ALD    & 0.505 (0.159) & 1.025 (0.251) & 0.259 (0.087) & 0.500 (0.154) & 0.236 (0.087) & 1 (0) & 0 (0) & 2.994 (0.404) \\
        hetero & 0.284 (0.087) & 0.637 (0.147) & 0.112 (0.038) & 0.277 (0.075) & 0.113 (0.044) & 1 (0) & 0 (0) & 2.817 (0.411) \\
        t3     & 0.337 (0.102) & 0.723 (0.166) & 0.142 (0.050) & 0.322 (0.085) & 0.146 (0.063) & 1 (0) & 0 (0) & 2.853 (0.391) \\
        contam & 0.353 (0.098) & 0.753 (0.174) & 0.156 (0.056) & 0.328 (0.076) & 0.159 (0.044) & 1 (0) & 0 (0) & 2.901 (0.361) \\
        cauchy & 119.5 (1062.8)
               & \num{2.0e5} (\num{2.0e6})
               & \num{2.8e4} (\num{2.7e5})
               & \num{2.1e5} (\num{2.1e6})
               & \num{3.3e4} (\num{3.3e5})
               & 1 (0)
               & 0 (0)
               & 42.63 (333.9) \\
        \bottomrule
    \end{tabular}
\end{table}

\begin{table}[htbp]
    \centering
    \caption{Monte Carlo result for quantile generalized additive models}
    \label{tab:model_QGAM}
    \setlength{\tabcolsep}{4pt}
    \renewcommand{\arraystretch}{1.1}
    \small
    \begin{tabular}{lccccccccc}
        \toprule
               &               &                    \multicolumn{4}{c}{MSE}                    &       &       &               \\
        \cmidrule(lr){3-6}
        Error  & PE$_{\theta}$ & $X_1$         & $X_2$         & $X_3$         & $X_4$         & Sens. & Spec. & CL            \\
        \midrule
        normal & 0.248 (0.060) & 1.480 (0.152) & 0.019 (0.015) & 0.121 (0.036) & 0.020 (0.016) & 1 (0) & 0 (0) & 2.992 (0.381) \\
        ALD    & 0.427 (0.140) & 1.582 (0.182) & 0.062 (0.048) & 0.246 (0.094) & 0.060 (0.050) & 1 (0) & 0 (0) & 3.018 (0.428) \\
        hetero & 0.230 (0.068) & 1.481 (0.147) & 0.016 (0.015) & 0.118 (0.035) & 0.017 (0.016) & 1 (0) & 0 (0) & 2.971 (0.420) \\
        t3     & 0.270 (0.076) & 1.484 (0.144) & 0.022 (0.019) & 0.134 (0.044) & 0.023 (0.020) & 1 (0) & 0 (0) & 2.969 (0.405) \\
        contam & 0.257 (0.063) & 1.508 (0.156) & 0.022 (0.017) & 0.129 (0.035) & 0.021 (0.019) & 1 (0) & 0 (0) & 3.001 (0.352) \\
        cauchy & 0.578 (1.299) & 3.076 (11.19) & 0.229 (1.675) & 1.309 (9.401) & 0.317 (2.566) & 1 (0) & 0 (0) & 5.814 (38.61) \\
        \bottomrule
    \end{tabular}
\end{table}
\end{landscape}

\begin{landscape}
\begin{table}[H]
    \centering
    \caption{Monte Carlo result for spatially varying coefficient model}
    \label{tab:model_CSVC}
    \setlength{\tabcolsep}{4pt}
    \renewcommand{\arraystretch}{1.1}
    \small
    \begin{tabular}{lccccccccc}
        \toprule
               &               &                    \multicolumn{4}{c}{MSE}                    &       &       &               \\
        \cmidrule(lr){3-6}
        Error  & PE$_{\theta}$ & $X_1$         & $X_2$         & $X_3$         & $X_4$         & Sens. & Spec. & CL            \\
        \midrule
        normal & 0.272 (0.066) & 1.454 (0.142) & 0.025 (0.020) & 0.129 (0.042) & 0.026 (0.020) & 1 (0) & 0 (0) & 2.989 (0.377) \\
        ALD    & 0.467 (0.151) & 1.595 (0.193) & 0.077 (0.052) & 0.274 (0.108) & 0.075 (0.056) & 1 (0) & 0 (0) & 3.026 (0.421) \\
        t3     & 0.316 (0.091) & 1.488 (0.145) & 0.035 (0.035) & 0.171 (0.098) & 0.041 (0.049) & 1 (0) & 0 (0) & 2.962 (0.410) \\
        contam & 0.333 (0.088) & 1.502 (0.146) & 0.039 (0.036) & 0.164 (0.048) & 0.036 (0.028) & 1 (0) & 0 (0) & 2.987 (0.365) \\
        cauchy & 114.8 (1016.2)
               & \num{4.9e5} (\num{4.9e6})
               & \num{7.0e4} (\num{7.0e5})
               & \num{5.0e5} (\num{5.0e6})
               & \num{8.4e4} (\num{8.4e5})
               & 1 (0)
               & 0 (0)
               & 47.69 (371.2) \\
        \bottomrule
    \end{tabular}
\end{table}
\begin{table}[H]
    \centering
    \caption{Monte Carlo result for SS-SVCQR}
    \label{tab:model_SVCQR}
    \setlength{\tabcolsep}{4pt}
    \renewcommand{\arraystretch}{1.1}
    \small
    \begin{tabular}{lccccccccc}
        \toprule
               &               &                    \multicolumn{4}{c}{MSE}            &       &       &               \\
        \cmidrule(lr){3-6}
        Error  & PE$_{\theta}$ & $X_1$         & $X_2$         & $X_3$         & $X_4$ & Sens. & Spec. & CL            \\
        \midrule
        normal & 0.098 (0.027) & 0.195 (0.014) & 0 (0) & 0.159 (0.010) & 0 (0) & 1 (0) & 1 (0)         & 0.474 (0.009) \\
        ALD    & 0.193 (0.054) & 0.405 (0.035) & 0 (0) & 0.341 (0.038) & 0 (0) & 1 (0) & 0.985 (0.086) & 1.075 (0.023) \\
        hetero & 0.086 (0.021) & 0.182 (0.012) & 0 (0) & 0.150 (0.009) & 0 (0) & 1 (0) & 1 (0)         & 0.494 (0.012) \\
        t3     & 0.113 (0.028) & 0.234 (0.016) & 0 (0) & 0.194 (0.012) & 0 (0) & 1 (0) & 1 (0)         & 0.630 (0.015) \\
        contam & 0.107 (0.032) & 0.226 (0.016) & 0 (0) & 0.186 (0.015) & 0 (0) & 1 (0) & 1 (0)         & 0.634 (0.018) \\
        cauchy & 0.156 (0.041) & 0.333 (0.030) & 0 (0) & 0.277 (0.030) & 0 (0) & 1 (0) & 0.990 (0.070) & 6.611 (21.23) \\
        \bottomrule
    \end{tabular}
\end{table}
\end{landscape}

Tables~\ref{tab:model_QR}--\ref{tab:model_SVCQR} summarize Monte Carlo means and standard deviations of the above metrics in $R=100$ replicates in Scenario~2.  For the purely global quantile regression (QR) in Table~\ref{tab:model_QR}, the global parameter error $\mathrm{PE}_\theta$ and the deviation MSEs are substantially larger than those of the spatial competitors whenever there is spatial nonstationarity.  Moreover, QR classifies all covariates as global, leading to a sensitivity essentially equal to zero and a specificity equal to one, and its predictive performance, measured by $\mathrm{CL}_{\mathrm{test}}$, is clearly inferior once the error distribution departs from the Gaussian case.

Standard GWR and mixed GWR (Tables~\ref{tab:model_GWR} and \ref{tab:model_MGWR}) reduce $\mathrm{PE}_\theta$ and the deviation MSEs relative to QR, but do so by treating almost all effects as spatially varying.  As a consequence, sensitivity is essentially one, while specificity is close to zero across all error distributions, indicating severe overfitting of local structure.  QGAM (Table~\ref{tab:model_QGAM}) shows a similar pattern: achieves smaller $\mathrm{PE}_\theta$ and deviation MSEs than QR, but again classifies nearly all covariates as local, with sensitivity near one and specificity near zero.  The classical SVC model in Table~\ref{tab:model_CSVC} behaves analogously, and under heavy-tailed Cauchy errors, its parameter error and deviation MSEs can become extremely large despite reasonably small check loss in some scenarios.

In contrast, the proposed SS--SVCQR estimator in Table~\ref{tab:model_SVCQR} delivers uniformly small global parameter errors and deviation MSEs across all error settings.  With the larger penalty used here, SS--SVCQR achieves sensitivity essentially equal to one and specificity close to one for normal, heteroscedastic, $t_3$, and contaminated-normal errors, and only slightly below one for ALD and Cauchy errors. Thus, SS--SVCQR is the only method that simultaneously recovers both truly local covariates $\{X_1,X_3\}$ and truly global covariates $\{X_2,X_4\}$ with high precision. Its predictive check loss $\mathrm{CL}_{\mathrm{test}}$ is uniformly lower than, or comparable to, that of the best competitor for light- and moderately heavy-tailed errors, and remains stable even under Cauchy errors, where several fully spatial methods suffer from dramatic loss of efficiency.  These results suggest that SS--SVCQR achieves near-oracle estimation of global coefficients, reliably recovers local structure, and maintains competitive predictive performance over a wide range of error distributions in Scenario~2.

Monte Carlo averages and standard deviations of these metrics show that SS--SVCQR has high sensitivity and specificity for distinguishing local (varying) from global (constant) covariates, even under heavy-tailed and heteroskedastic errors; the deviation MSEs remain small relative to the signal variance, and the global coefficients are estimated with near-oracle precision.  Predictive performance is competitive with, or better than, fully spatially varying quantile regression and markedly improves over a purely global quantile regression when spatial nonstationarity is present.

\section{Real Data}\label{sec:realdata}

We illustrate the proposed method using the Lucas County house price data, a benchmark hedonic dataset that has been widely used in spatial econometrics and varying-coefficient modeling. The data contain 25,357 single-family home sales recorded by the Lucas County auditor between 1993 and 1998, together with structural characteristics and projected coordinates in the Ohio North State Plane system (EPSG:2834). Figure \ref{fig:log_price} maps the log-sale price for all transactions within the administrative boundary of Lucas County\cite{spData, spDataLarge}.

\begin{figure}[ht]
    \centering
    \includegraphics[width=1\linewidth]{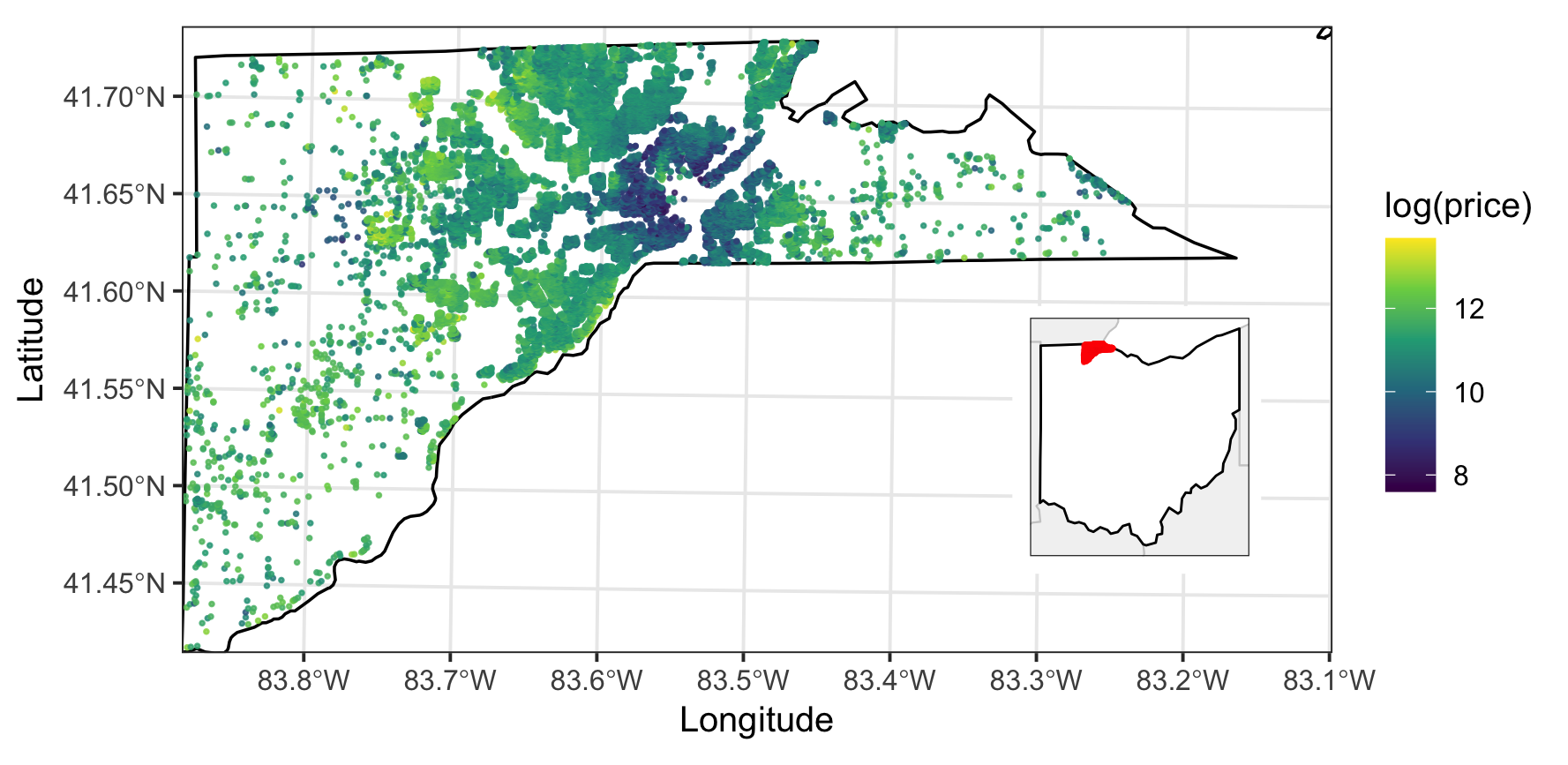}
    \caption{Spatial distribution of log house prices in Lucas County, Ohio. \textit{Points show 25,357 single-family transactions recorded by the Lucas County auditor between 1993 and 1998 (\textit{spData} \textit{house} dataset). Colours indicate the log sale price, with darker shades corresponding to cheaper properties and yellow–green tones to more expensive ones. The polygon outlines the administrative boundary of Lucas County. Coordinates are shown as longitude and latitude; the original data are provided in the Ohio North State Plane projection (EPSG:2834).}}
    \label{fig:log_price}
\end{figure}

The map reveals pronounced spatial heterogeneity in house values. A corridor of relatively low-priced properties appears in the central urban core along the Maumee River, whereas higher prices concentrate in the western and northwest suburbs and in a few pockets along the eastern edge of the county. The gradient is not purely radial: there are inexpensive enclaves embedded in otherwise high-priced neighborhoods and vise versa, suggesting that location interacts with local housing characteristics in a complex way rather than through a simple distance-to-CBD effect.

These patterns suggest that the marginal effects of housing attributes are unlikely to be spatially constant. For example, additional floor area may command a larger premium in high-amenity western suburbs than in older inner-city neighborhoods, and age or garage size may matter differently across locations. In addition, urban housing markets are known to be skewed and heavy-tailed, making least-squares specifications sensitive to outliers. To capture both spatially varying relationships and distributional heterogeneity, we apply our sparse–smooth spatially varying-coefficient quantile regression to model selected structural covariates with local coefficients, while keeping other controls as global effects.

Let $Y_i$ denote the sale price of house $i$ and $u_i\in\mathcal U\subset\mathbb R^2$ its location. Following the hedonic literature on Lucas County, we work with the log–price $Y_i^\ast = \log(Y_i)$, which stabilizes the right–skewed price distribution and allows the regression coefficients to be interpreted as semi–elasticities. The original coordinates are transformed to longitude–latitude and then linearly rescaled to $[0,1]^2$ for numerical stability in the graph construction described in Section~\ref{sec:estimation}.

Consistent with previous hedonic studies on this data set, we focus on structural attributes of the dwelling as the main drivers of price differences: building age, interior living area, number of rooms and bedrooms, and garage size; see, for example, the specifications used in spatial autoregressive and geographically weighted models. In addition, we control for building type (number of stories, exterior wall material, garage type) and for sale year dummies that capture market-wide temporal trends.

In the notation of Section~\ref{sec:model}, we set $Y_i=Y_i^\ast$ and partition the covariates into a purely global block $\bm Z_i$ and a potentially spatially varying block $\bm X_i$ as follows. The global block collects a constant, structural–type indicators and time dummies
\begin{equation*}
    \bm Z_i
    =
    \big(1,
    \text{stories}_i,
    \text{wall}_i,
    \text{garageType}_i,
    \text{syear}_i,
    \text{baths}_i,
    \text{halfbaths}_i\big)^\top,
\end{equation*}
where $\text{stories}_i$, $\text{wall}_i$, $\text{garageType}_i$ and $\text{syear}_i$ are encoded as sets of dummy variables. These effects are interpreted as global in the sense that the impact of being, say, a two–story brick house or selling in 1997 is assumed not to vary smoothly across the county.

The potentially spatially varying block contains continuous structural characteristics for which the hedonic literature reports marked submarket heterogeneity:
\begin{equation*}
    \bm X_i =
    \big(\text{age}_i, \log(\text{TLA}_i), \log(\text{lotsize}_i),
    \text{rooms}_i, \text{beds}_i, \text{garageArea}_i\big)^\top.
\end{equation*}
Here $\text{age}_i$ is the age of the house at sale, $\text{TLA}_i$ is the total living area (square feet), $\text{rooms}_i$ and $\text{beds}_i$ are the number of rooms and bedrooms, and $\text{garageArea}_i$ is the garage floor area (with zero for houses without a garage). All components of $\bm X_i$ are centered and scaled to unit variance before fitting so that the group penalty acts comparably across predictors. In the SS--SVCQR formulation, each coefficient is decomposed as
\[
\beta_j(u) = \beta_{G,j} + \delta_j(u),
\quad j=1,\dots,p,
\]
where $\beta_{G,j}$ represents the global (location–invariant) marginal effect of predictor $X_{ij}$ on the $\tau$th conditional quantile of log–price, and $\delta_j(u)$ captures a spatial deviation field around this global level. The adaptive group penalty in~\eqref{eq:obj-main} decides, for each of these structural variables, whether $\delta_j(\cdot)$ should be shrunk to zero (globally constant effect) or retained as a smoothly varying coefficient.

Table~\ref{tab:lucas-vars} summarizes the variables used in the analysis and their roles in the model.
\begin{landscape}
\begin{table}[t]
\centering
\caption{Variables used in the Lucas County housing application and their roles in the SS--SVCQR model. SVC stands for spatially varying coefficient.}
\label{tab:lucas-vars}
\begin{tabular}{llll}
\hline
Symbol / name & Data field & Description & Role in model \\
\hline
$Y$ & \texttt{price} & Sale price (USD), log--transformed & Response \\
$u_1,u_2$ & coordinates & House location, rescaled to $[0,1]^2$ & Spatial index \\
\hline
\multicolumn{4}{l}{\emph{Global covariates $\bm Z_i$}}\\
Intercept & -- & Constant term & Global effect \\
\texttt{stories} & \texttt{stories} & Building type (number of stories, factor) & Global effect \\
\texttt{wall} & \texttt{wall} & Exterior wall material (factor) & Global effect \\
\texttt{garageType} & \texttt{garage} & Garage type (factor) & Global effect \\
\texttt{syear} & \texttt{syear} & Sale year 1993--1998 (factor) & Global effect \\
\texttt{baths} & \texttt{baths} & Number of full bathrooms & Global effect \\
\texttt{halfbaths} & \texttt{halfbaths} & Number of half bathrooms & Global effect \\
\hline
\multicolumn{4}{l}{\emph{Potentially spatially varying covariates $\bm X_i$}}\\
$X_1$ & \texttt{age} & House age at sale (years), standardized & SVC candidate \\
$X_2$ & \texttt{TLA} & Total living area, log and standardized & SVC candidate \\
$X_3$ & \texttt{lotsize} & Lot size, log and standardized & SVC candidate \\
$X_4$ & \texttt{rooms} & Total number of rooms, standardized & SVC candidate \\
$X_5$ & \texttt{beds} & Number of bedrooms, standardized & SVC candidate \\
$X_6$ & \texttt{garageArea} & Garage area (sq ft), standardized & SVC candidate \\
\hline
\end{tabular}
\end{table}
\end{landscape}
Our choice of which regressors are allowed to vary over space is guided by both the hedonic housing literature on the Lucas County dataset and more recent studies using spatially varying coefficient (SVC) or multiscale geographically weighted regression models.

For Lucas County data, spatial econometric and semiparametric analyzes such as LeSage and Pace (2009, Chap.~7), Zhu et al. (2011), Mínguez et al. (2013), and Basile et al. (2014) consistently show that prices are mainly driven by structural characteristics: floor area, lot size, number of rooms and bedrooms, and age, together with strong spatial trend components. In contrast, design and transaction dummies (wall material, storey, garage type, year-of-sale dummies) typically enter as conventional fixed effects whose estimated coefficients are comparatively stable across space. These findings suggest that the marginal willingness to pay for size and vintage interacts more strongly with location than the effects of building style or sale year.

More broadly, recent SVC and multiscale GWR applications to hedonic pricing problems report a similar pattern. Murakami et al. (2017) and Dambon et al. (2022) document pronounced spatial variation in the coefficients of structural attributes such as floor area and age, while some categorical or regulatory variables behave almost globally. Chica-Olmo et al. (2019) and Hong et al. (2022) likewise find that structural and locational amenities often require flexible spatial effects, whereas purely institutional or temporal indicators can be treated as global without substantial loss of fit. Empirical work on housing and land prices using mixed GWR or related frameworks reaches similar conclusions, with some coefficients operating at near-global scales and others varying at very local scales (e.g., Fotheringham, Yang and Kang, 2017; Muto et al., 2023).

Motivated by this evidence, in our SS--SVCQR specification, we treat the main structural characteristics---age, total living area, lot size, number of rooms and bedrooms, and garage size---as candidate spatially varying coefficients. These variables represent the quantity and effective quality of housing services, whose marginal valuation is known to depend on local neighborhood conditions and submarkets. In contrast, the intercept, the design-related dummies (storey, wall material, garage type), the number of full and half bathrooms, and the sale-year dummies enter as global coefficients. They capture broad market conditions, building regulations, and style differences that are expected to be more homogeneous within Lucas County. Importantly, the group penalty in our estimator can still shrink a \textbf{candidate} spatially varying effect back to a global one if the data do not support substantial spatial heterogeneity.

To assess out-of-sample performance, we randomly partition the 25,357 transactions into a training set (80\%) and a test set (20\%). The split is performed at the observation level, stratified by sale year to preserve the temporal composition of the sample across folds. All tuning parameters for SS--SVCQR (the global–local penalty $\lambda$ and the graph-smoothing parameter) are selected on the training set by five-fold cross-validation, minimizing the average check loss at the target quantile. Competing models (Global QR, QGAM, and GWR) are tuned using analogous procedures: for QGAM we choose the effective degrees of freedom via REML, while for GWR we select the bandwidth by cross-validation. The reported test losses and pseudo-$R^2$ are computed on the held-out test set.

\begin{table}[t]
    \centering
    \caption{Out-of-sample performance at the median ($\tau=0.5$) in Lucas County. The metric $\mathrm{CL}_{0.5}$ is the average check loss on the test set; $R^2_{0.5}$ is a pseudo-$R^2$ relative to a constant-quantile null model.}
    \label{tab:lucas-model-comp}
    \begin{tabular}{lccc}
        \toprule
        Model & $\mathrm{CL}_{0.5}$ & $R^2_{0.5}$ & Improvement over Global QR \\
        \midrule
        SS--SVCQR             & 0.1310 & \texttt{0.64} & $+53.2\%$ \\
        Global QR             & 0.2799 & \texttt{0.49} & --- \\
        QGAM                  & 0.1464 & \texttt{0.50} & $+47.6\%$ \\
        GWR                   & 2.4878 & \texttt{-7.51} & strongly negative \\
        \bottomrule
    \end{tabular}
\end{table}
In the median, SS--SVCQR achieves the lowest out-of-sample check loss among all competitors (Table~\ref{tab:lucas-model-comp}). Compared to the global QR benchmark, SS–SVCQR reduces the median check loss by approximately 53.2\%. Compared to the flexible QGAM baseline, the reduction is approximately 9.5\%. The GWR fit performs poorly in this setting, with substantially inflated check loss and near-zero pseudo-$R^2$, reflecting its instability under local collinearity and the lack of explicit regularization.

\begin{table}[t]
    \centering
    \caption{Estimated global coefficients and spatial deviations for the SS--SVCQR
    model at $\tau=0.5$.  The column $\hat\beta_{G,j}$ reports the global effect,
    $\|\hat\delta_j\|_2$ is the $\ell_2$-norm of the deviation field, and
    $[\min\hat\delta_j,\max\hat\delta_j]$ summarizes its range over sampled
    locations.}
    \label{tab:lucas-global-local}
    \begin{tabular}{lccc}
    \toprule
    Covariate $X_j$ & $\hat\beta_{G,j}$ & $\|\hat\delta_j\|_2$ &
    $[\min\hat\delta_j,\max\hat\delta_j]$ \\
    \midrule
    House age (years)        & $-0.271$ & 0.168 &
    $[-1.05,\ 1.16]$ \\
    Floor area (log TLA)     & $0.244$  & 0.036 &
    $[-0.27,\ 0.38]$ \\
    Number of rooms          & $-0.031$ & 0.058 &            $[-0.57,\ 0.72]$ \\
    Number of bedrooms       & $0.0067$ & $8.4\times10^{-7}$ & $[-7.2\times10^{-6}, \ 6.1\times10^{-5}]$ \\
    Garage area (sq.~ft.)    & $-0.032$ & 0.100 &
    $[-1.07,\ 0.85]$ \\
    Lot size (log)           & $0.175$  & 0.110 &
    $[-0.78,\ 0.99]$ \\
    \bottomrule
    \end{tabular}
\end{table}
Table~\ref{tab:lucas-global-local} summarizes the decomposition of each continuous predictor into a global effect and a spatial deviation. The signs and magnitudes of the global coefficients are broadly in line with hedonic expectations and prior evidence: larger floor area and lot size are associated with higher prices, whereas older dwellings tend to be discounted. 

The group-wise deviation norms reveal that bedrooms are essentially a purely global predictor: $\|\hat\delta_{\text{beds}}\|_2$ is numerically negligible. By contrast, age, garage area, and lot size exhibit substantial spatial heterogeneity, with deviation ranges on the order of $\pm 1$ on the log price scale. Floor area behaves as a predominantly global covariate with only mild spatial variation, while the number of rooms shows intermediate behaviour.

\begin{figure}[ht]
    \centering
    \includegraphics[width=\textwidth]{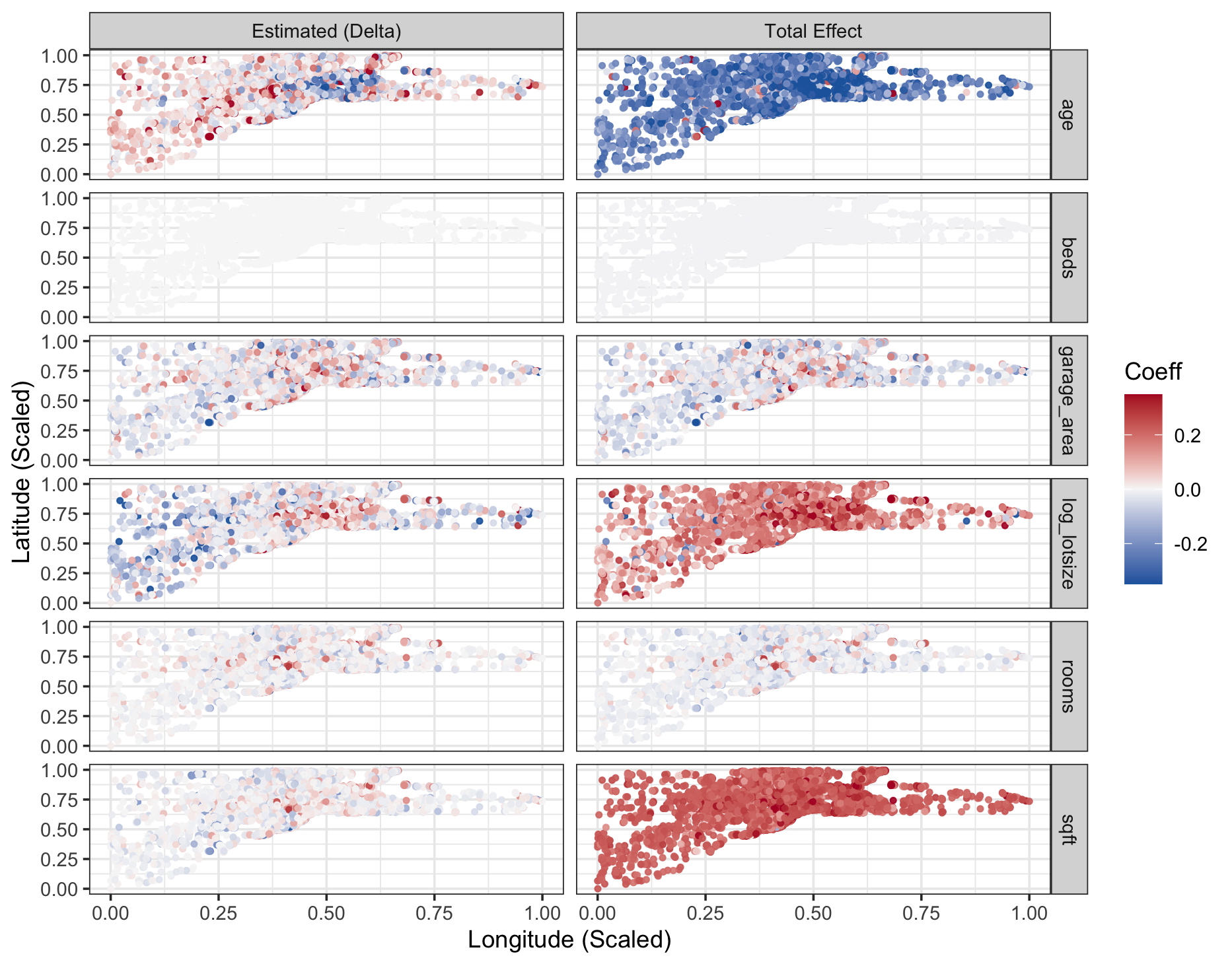}
    \caption{Estimated spatial deviations (left column) and total effects
    $\hat\beta_{G,j}+\hat\delta_j(u)$ (right column) at $\tau=0.5$ for the
    SS--SVCQR model. Each row corresponds to one covariate; colours indicate the
    local coefficient magnitude.}
    \label{fig:lucas-svcqr}
\end{figure}
Figure~\ref{fig:lucas-svcqr} visualizes the estimated deviation fields $\hat\delta_j(u)$ (left panels) and the corresponding total effects $\hat\beta_{G,j}+\hat\delta_j(u)$ (right panels). 

The coefficient on dwelling age is negative on average but displays marked spatial heterogeneity. In a broad central–northern band of the county, characterized by an older housing stock, the estimated marginal effect of one additional year of age is strongly negative, indicating substantial depreciation of older units relative to newer construction. Towards several peripheral neighbourhoods, the effect becomes less negative and in some pockets slightly positive, consistent with historic or high-amenity properties commanding a premium once size and other attributes are held fixed.

The effect of floor area behaves predominantly global. The deviation surface for log TLA is comparatively smooth and of small amplitude relative to its global component, implying that an incremental increase in floor area is capitalized into house prices in a fairly homogeneous manner over Lucas County. This finding is consistent with previous hedonic studies and meta-analyzes that report a stable and strongly positive association between square footage and sale price.

The coefficients on garage area exhibit moderate spatial variation around a relatively small global effect. In most locations the marginal value of additional garage space is close to zero, with only modest positive premia in some auto-oriented neighbourhoods and slightly negative effects in others. This suggests that, conditional on overall floor area and the presence of a garage, extra garage square footage plays a secondary role in this market compared with indoor living space and lot size. By contrast, the lot-size premium is clearly positive across almost the entire county, with stronger effects in several lower-density suburban areas and somewhat attenuated effects in denser neighbourhoods. This pattern is in line with urban economic arguments about the trade-off between land consumption and accessibility: larger lots are valued everywhere, but their marginal contribution to price is amplified where land is relatively abundant.

Floor area (log total living area) behaves as a predominantly global covariate with a strong positive effect that varies only mildly over space. The number of rooms shows intermediate behaviour: it is positively associated with price overall, but its marginal impact fluctuates across submarkets once floor area is controlled for. In contrast, the effect of the number of bedrooms is essentially spatially invariant, with a negligible deviation field, indicating that the valuation of an additional bedroom—conditional on total living area and rooms—can be well approximated by a single global coefficient.

\begin{table}[ht]
  \centering
  \caption{Out-of-sample quantile prediction and residual diagnostics for the Lucas County housing application.}
  \label{tab:lucas-multi-quant-moran}
  \small
  \begin{tabular}{lcccccccc}
        \toprule
        & \multicolumn{6}{c}{Out-of-sample performance} &
          \multicolumn{2}{c}{Diagnostics} \\
        \cmidrule(lr){2-7}\cmidrule(lr){8-9}
        Model
        & $\mathrm{CL}_{0.25}$ & $R^2_{0.25}$
        & $\mathrm{CL}_{0.5}$  & $R^2_{0.5}$
        & $\mathrm{CL}_{0.75}$ & $R^2_{0.75}$
        & Moran's $I$ & $p$-value \\
        \midrule
        SS--SVCQR & 0.1262 & 0.6132 & 0.1310 & 0.6455 & 0.0999 & 0.6471 & 0.126 & 0 \\
        Global QR & 0.1326 & 0.4882 & 0.1448 & 0.4973 & 0.1026 & 0.5348 & 0.448 & 0 \\
        QGAM      & 0.1382 & 0.4667 & 0.1482 & 0.4856 & 0.1050 & 0.5235 & 0.402 & 0 \\
        GWR       & 2.1234 & -7.193 & 2.0732 & -6.196 & 2.0231 & -8.175 & 0.028 & 0 \\
        \bottomrule
  \end{tabular}
\end{table}
Taken together, the results in Table~\ref{tab:lucas-multi-quant-moran} show that SS--SVCQR provides the best overall fit to the Lucas County data: it consistently improves out-of-sample quantile prediction relative to both global and semiparametric benchmarks, while also reducing residual spatial structure compared with purely global specifications. Combined with the coefficient patterns documented above, this evidence suggests that allowing key structural attributes to have sparse, smoothly varying effects over space is an effective and interpretable way to capture submarket heterogeneity in large housing datasets.

\section{Discussion}\label{sec:discussion}

In this paper, we have introduced a sparse--smooth spatially varying coefficient quantile regression (SS--SVCQR) framework. By integrating quantile loss with graph Laplacian regularization and group-sparse penalization, our method addresses the dual challenges of robust estimation under non-Gaussian tails and interpretable structural identification in spatial modeling.

The primary contribution of this work lies in the explicit decomposition of coefficients into global baselines and local deviations. Unlike standard Geographically Weighted Regression (GWR) \cite{Fotheringham2002}, which often yields coefficients that fluctuate everywhere due to local sampling variability, our approach employs a group $L_2$ penalty to perform automated global-local model selection. This \textit{parsimony-in-complexity} allows practitioners to distinguish between predictors that exert a universal influence across the domain and those driven by local contextual factors. Furthermore, the use of the normalized graph Laplacian $\bm L_{\mathrm{sym}}$ provides a robust mechanism to handle irregular sampling designs, preventing the overfitting often observed in kernel-based methods when data density is low.

Computationally, the convex formulation of SS--SVCQR enables the use of efficient ADMM and smoothed proximal-gradient algorithms. Our theoretical analysis confirms that this computational efficiency does not come at the cost of statistical rigor, as the estimator enjoys root-$n$ consistency and oracle selection properties under suitable conditions. From a practical perspective, these properties justify the use of SS--SVCQR as a drop-in replacement for standard QR in applications where spatial nonstationarity and heavy-tailed responses are of concern.

Despite these advantages, several limitations warrant further investigation. First, the performance of the spatial regularization relies on the construction of the adjacency graph. Although we utilized a mutual $k$-nearest-neighbor graph to adapt to density variations, the choice of $k$ and the bandwidth $\sigma$ remains a hyperparameter selection problem. Although cross-validation mitigates this, data-driven graph learning techniques \cite{Hallac2015} could potentially produce geometries that better reflect the underlying manifold of the spatial process. Second, our current framework focuses on a fixed quantile level $\tau$. In applications requiring a comprehensive view of the conditional distribution, estimating multiple quantiles simultaneously with non-crossing constraints \cite{HePanTanZhou2023} would be a valuable extension. Finally, while the graph Laplacian enforces smooth variations, it penalizes abrupt changes (e.g., discontinuities across administrative boundaries or physical barriers). Incorporating edge-preserving penalties, such as Total Variation in graphs \cite{Wang2016}, could allow the model to capture piecewise constant spatial regimes more effectively.

Future work will also explore spatiotemporal extensions. By constructing a product graph over space and time, the proposed regularization framework can be naturally generalized to capture dynamic heterogeneity. Additionally, developing post-selection inference tools to quantify uncertainty for the identified local deviations remains an important open direction.

\appendix
\section{Proofs of Theorems}\label{app:proofs}

Throughout this appendix, we keep the notation of Sections~\ref{sec:model}--\ref{sec:theory}.
For a vector $\bm v\in\mathbb R^n$ we write
\begin{equation*}
    \|\bm v\|_2^2 = \sum_{i=1}^n v_i^2, \quad
    \langle\bm a, \bm b\rangle = \bm a^\top\bm b  
\end{equation*}
for the Euclidean norm and inner product, and use $C,c,c_0,c_1,\dots$ for generic finite positive constants whose value may change from line to line. The notation $\preceq$ denotes the L\"owner order on symmetric matrices. Recall the diagonal multipliers $\bm X_{\odot j}=\diag(X_{1j},\dots,X_{nj})$, the (symmetric normalized) Laplacian $\bm L_{\mathrm{sym}}$ in~\eqref{eq:Lsym}, the degree matrix $\bm D$, and the degree--weighted centering projector $\mathsf{Proj}_D$ in~\eqref{eq:proj}.The subgradient of the check loss is $\psi_\tau(r)=\tau-\mathbf 1\{r<0\}$ and we stack it as $\bm\psi_\tau(\bm r)=(\psi_\tau(r_1),\dots,\psi_\tau(r_n))^\top$.

We record two standard facts that will be used repeatedly.

\paragraph{Knight identity and local quadratic expansion.}
For any $r,t\in\mathbb R$,
\begin{equation}\label{eq:knight}
    \rho_\tau(r - t) - \rho_\tau(r)
    = -t \psi_\tau(r)
    + \int_0^t \left\{\mathbf 1(r\le s)-\mathbf 1(r\le 0)\right\}\mathrm{d}s.
\end{equation}
This is the well–known identity of \cite{Knight1998}. For completeness, we briefly indicate the derivation. The function $u\mapsto \rho_\tau(u)$ is piecewise linear with a kink at $0$, and its derivative exists Lebesgue-a.e. and is given by $\psi_\tau(u)=\tau-\mathbf 1\{u<0\}$. Integrating the derivative of $s\mapsto \rho_\tau(r-s)$ between $0$ and $t$ then yields~\eqref{eq:knight}; see also \cite[Sec.~4.3]{Koenker2005}.

Let $\varepsilon_i=Y_i-Q_\tau(Y_i\mid \bm Z_i,\bm X_i,u_i)$ be as in Assumption~\ref{ass:A3}, and $f_{\varepsilon\mid Z,X,U}(\cdot)$ the corresponding conditional density. By \eqref{eq:knight}, conditioning on $(\bm Z_i,\bm X_i,u_i)$ and taking expectations gives
\begin{align*}
    \mathrm E\{\rho_\tau(\varepsilon_i-t)-\rho_\tau(\varepsilon_i)\mid \bm Z_i,\bm X_i,u_i\}
    & =
    - t \mathrm E\{\psi_\tau(\varepsilon_i)\mid \bm Z_i,\bm X_i,u_i\}\\
    & \quad
    + \int_0^t \Big(
    \Pr(\varepsilon_i\le s\mid \bm Z_i,\bm X_i,u_i)
    -\Pr(\varepsilon_i\le 0\mid \bm Z_i,\bm X_i,u_i)
    \Big) \mathrm{d}s.
\end{align*}
By the defining property of conditional quantiles, $\Pr(\varepsilon_i\le 0\mid \bm Z_i,\bm X_i,u_i)=\tau$ and $\mathrm E\{\psi_\tau(\varepsilon_i)\mid \bm Z_i,\bm X_i,u_i\}=0$. Moreover, by the continuity of the conditional density at $0$ (Assumption~\ref{ass:A3})
\begin{equation*}
    \Pr(\varepsilon_i\le s\mid \bm Z_i,\bm X_i, u_i)
    - \Pr(\varepsilon_i\le 0\mid \bm Z_i, \bm X_i, u_i)
    = f_{\varepsilon\mid Z, X, U}(0)s + o(s)
\end{equation*}
uniformly for $|s|\le \eta$ and some $\eta>0$.
Hence, for $|t|\le \eta$,
\begin{align}\label{eq:lqa}
    \mathrm E\left\{\rho_\tau(\varepsilon_i-t) - \rho_\tau(\varepsilon_i)\mid \bm Z_i, \bm X_i, u_i\right\}
    &= \int_0^t \left\{f_{\varepsilon\mid Z, X, U}(0)s + o(s)\right\} \mathrm{d}s
    \nonumber \\
    &= \frac12 f_{\varepsilon\mid Z,X,U}(0) t^2 + o(t^2),
\end{align}
uniformly in $(\bm Z_i,\bm X_i,u_i)$.
Summing over $i$ leads to a local quadratic approximation of the empirical loss around the true residuals; cf.\ Assumption~\ref{ass:A7} and \cite[Sec.~4.3]{Koenker2005}.

\paragraph{Graph Poincar\'e inequality on the centered subspace.}
Let $\mathcal S\subset\mathbb R^n$ be the subspace
\begin{equation*}
    \mathcal S
    =\{\bm v\in\mathbb R^n:\ \bm 1_{\mathcal C}^\top \bm D\bm v=0\ \text{for every connected component }\mathcal C\}.
\end{equation*}
By construction, the null space of $\bm L_{\mathrm{sym}}$ on a connected component $\mathcal C$ is spanned by $\bm D^{1/2}\bm 1_{\mathcal C}$, and therefore $\mathcal S$ is orthogonal, in the $\bm D$–inner product, to $\ker(\bm L_{\mathrm{sym}})$.
Since $\bm L_{\mathrm{sym}}$ is symmetric positive semidefinite with eigenvalues in $[0,2]$, there exists a constant $c_\lambda>0$ (depending only on the graph sequence and bounded away from $0$ under Assumption~\ref{ass:A1}) such that
\begin{equation}\label{eq:poincare}
    \bm v^\top \bm L_{\mathrm{sym}}\bm v  \ge c_\lambda \|\bm v\|_2^2,
    \quad \forall \bm v\in\mathcal S.
\end{equation}
This is a discrete Poincar\'e inequality; see, for instance, \cite[Chap.~1]{Chung1997} and \cite[Sec.~2]{vonLuxburg2007}.

Let $(\hat{\bm\theta}_P,\hat{\bm\Theta}_L)$ be any global minimizer of the penalized criterion $\mathcal L_n$ in Section~\ref{sec:theory}. Write the true parameter as $(\bm\theta_{P0},\bm\Theta_{L0})$ with components $\bm\delta_{j0}$, and define
\begin{equation*}
    q_{\tau, i}(\bm\theta_P, \bm\Theta_L)
    = \bm Z_i^\top\bm\alpha + \bm X_i^\top\bm\beta_G
    + \sum_{j=1}^p X_{ij}\delta_j(u_i),
    \quad
    r_i(\bm\theta_P, \bm\Theta_L) = Y_i - q_{\tau, i}(\bm\theta_P, \bm\Theta_L).
\end{equation*}
By optimality,
\begin{equation*}
    \mathcal L_n(\hat{\bm\theta}_P, \hat{\bm\Theta}_L)
    \le \mathcal L_n(\bm\theta_{P0}, \bm\Theta_{L0}),
\end{equation*}
so that for any feasible $(\bm\theta_P, \bm\Theta_L)$
\begin{align}\label{eq:basic-ineq}
    0
    &\le \mathcal L_n(\bm\theta_P, \bm\Theta_L)
       -\mathcal L_n(\hat{\bm\theta}_P, \hat{\bm\Theta}_L)\nonumber\\
    &= \sum_{i=1}^n
       \Big\{
          \rho_\tau\big(r_i(\bm\theta_P, \bm\Theta_L)\big)
          -\rho_\tau\big(r_i(\hat{\bm\theta}_P, \hat{\bm\Theta}_L)\big)
       \Big\} \nonumber \\
    & \quad
     + \sum_{j=1}^p \lambda_{1n} w_{j, n}
          \big(\|\bm\delta_j\|_2 - \|\hat{\bm\delta}_j\|_2\big)
     + \sum_{j=1}^p \lambda_{2n}
          \big(\bm\delta_j^\top \bm L_{\mathrm{sym}}\bm\delta_j
               - \hat{\bm\delta}_j^\top\bm L_{\mathrm{sym}}\hat{\bm\delta}_j\big).
\end{align}
In particular, choosing $(\bm\theta_P,\bm\Theta_L)=(\bm\theta_{P0},\bm\Theta_{L0})$ yields the inequality we will work with.

In what follows we use the shorthand
\begin{equation*}
    \Delta\bm\alpha = \hat{\bm\alpha}-\bm\alpha_0,\quad
    \Delta\bm\beta_G = \hat{\bm\beta}_G-\bm\beta_{G0},\quad
    \Delta\bm\delta_j = \hat{\bm\delta}_j-\bm\delta_{j0},
\end{equation*}
and denote by $\bm h=\sum_{j=1}^p\bm X_{\odot j}\Delta\bm\delta_j\in\mathbb R^n$ the aggregate contribution of the deviation perturbations. Let $\bm r_0 = (\varepsilon_1,\dots,\varepsilon_n)^\top$ be the true residual vector, and set
\begin{equation*}
    t_i = q_{\tau,i}(\hat{\bm\theta}_P,\hat{\bm\Theta}_L)-q_{\tau,i}(\bm\theta_{P0},\bm\Theta_{L0})
    = \bm Z_i^\top\Delta\bm\alpha
    + \bm X_i^\top\Delta\bm\beta_G
    + \sum_{j=1}^p X_{ij}\Delta\delta_j(u_i),
\end{equation*}
so that $r_i(\hat{\bm\theta}_P, \hat{\bm\Theta}_L)=\varepsilon_i-t_i$ and $r_i(\bm\theta_{P0}, \bm\Theta_{L0})=\varepsilon_i$.

Applying Knight identity~\eqref{eq:knight} with $r=\varepsilon_i$ and $t=t_i$ and summing over $i$ shows that the difference of empirical losses appearing in~\eqref{eq:basic-ineq} can be written as
\begin{equation}\label{eq:knight-sum}
    \sum_{i=1}^n \Big\{\rho_\tau(\varepsilon_i-t_i)-\rho_\tau(\varepsilon_i)\Big\}
    = -\langle\bm\psi_\tau(\bm\varepsilon), \bm t\rangle
      + \sum_{i=1}^n \int_0^{t_i}
          \big\{\mathbf 1(\varepsilon_i\le s)-\mathbf 1(\varepsilon_i\le 0)\big\} \mathrm{d}s,
\end{equation}
where $\bm t=(t_1,\dots,t_n)^\top$. The second term on the right-hand side admits a quadratic expansion via~\eqref{eq:lqa} and Assumption~\ref{ass:A7}.

Recall the active set $\mathcal A_0=\{j:\bm\delta_{j0} = \bm 0\}$ and its complement $\mathcal A_0^c=\{1, \dots, p\}\setminus\mathcal A_0$.

\begin{proof}[Proof of Theorem~\ref{thm:delta-selection}(i)]
Fix a covariate index $j\in\mathcal A_0^c$. We control the deviation $\Delta\bm\delta_j$ by isolating its contribution in the basic inequality~\eqref{eq:basic-ineq} with $(\bm\theta_P, \bm\Theta_L)=(\bm\theta_{P0}, \bm\Theta_{L0})$. Using~\eqref{eq:knight-sum} and the local quadratic approximation~\eqref{eq:lqa}, we obtain
\begin{align} \label{eq:lqa-sum2}
    &\sum_{i=1}^n
       \Big\{
          \rho_\tau(\varepsilon_i-t_i)-\rho_\tau(\varepsilon_i)
       \Big\}
    \nonumber\\
    &\quad
    = -\langle\bm\psi_\tau(\bm\varepsilon),\bm t\rangle
     + \frac12\sum_{i=1}^n f_{\varepsilon\mid Z, X, U}(0) t_i^2
     + R_n(\bm t),
\end{align}
where $R_n(\bm t)=o_p(\|\bm t\|_2^2)$ uniformly in local neighborhoods
$\{\|\bm t\|_2\le c\sqrt n\}$ by Assumption~\ref{ass:A7}.
Since $p$ is fixed and the penalty sequences satisfy Assumption~\ref{ass:A5},
standard arguments for penalized quantile regression
(cf.\ \cite[Sec.~4.3]{Koenker2005}) imply
$\|\Delta\bm\alpha\|_2+\|\Delta\bm\beta_G\|_2=O_p(n^{-1/2})$.
Therefore the contribution of the parametric block to
$\bm t$ is $O_p(1)$ in $\ell_2$, and we can write
\[
    \bm t = \bm h + \bm u, 
    \quad
    \|\bm u\|_2 = O_p(1),
\]
with $\bm h=\sum_{j=1}^p\bm X_{\odot j}\Delta\bm\delta_j$ as above. Expanding $t_i^2$ and using bounded covariates (Assumption~\ref{ass:A2}) shows that the cross-term $\sum f(0) h_i u_i$ and the purely parametric term $\sum f(0)u_i^2$ in~\eqref{eq:lqa-sum2} are of smaller order $O_p(\|\bm h\|_2)$ and $O_p(1)$, respectively, compared to the leading quadratic term in $\bm h$.
Absorbing these contributions into the remainder, we obtain
\begin{equation}\label{eq:one-group2}
    \sum_{i=1}^n
       \Big\{
          \rho_\tau(\varepsilon_i-t_i)-\rho_\tau(\varepsilon_i)
       \Big\}
    =
    -\langle\bm\psi_\tau(\bm\varepsilon), \bm h\rangle
    + \frac12\bm h^\top\bm W\bm h
    + o_p(\|\bm h\|_2^2),
\end{equation}
where $\bm W=\diag\{f_{\varepsilon\mid Z, X, U}(0), \dots\}$ has diagonal entries uniformly bounded between $c_f$ and $C_f$ by Assumption~\ref{ass:A3}.

Plugging this expression and $(\bm\theta_P, \bm\Theta_L)=(\bm\theta_{P0}, \bm\Theta_{L0})$ into~\eqref{eq:basic-ineq}, and then isolating the $j$th block, we have
\begin{align} \label{eq:group-split}
    0
    &\le
    -\langle\bm\psi_\tau(\bm\varepsilon), \bm h\rangle
    + \frac12\bm h^\top\bm W\bm h
    + o_p(\|\bm h\|_2^2) \nonumber\\
    & \quad
    + \lambda_{1n}w_{j, n}
       \big(\|\bm\delta_{j0}\|_2-\|\hat{\bm\delta}_j\|_2\big)
    + \lambda_{2n}
       \big(\bm\delta_{j0}^\top\bm L_{\mathrm{sym}}\bm\delta_{j0}
            -\hat{\bm\delta}_j^\top\bm L_{\mathrm{sym}}\hat{\bm\delta}_j\big)
    \nonumber \\
    & \quad
    + \sum_{\ell\neq j}\lambda_{1n}w_{\ell,n}
       \big(\|\bm\delta_{\ell0}\|_2-\|\hat{\bm\delta}_\ell\|_2\big)
    + \sum_{\ell\neq j}\lambda_{2n}
       \big(\bm\delta_{\ell0}^\top\bm L_{\mathrm{sym}}\bm\delta_{\ell0}
            -\hat{\bm\delta}_\ell^\top\bm L_{\mathrm{sym}}\hat{\bm\delta}_\ell\big).
\end{align}
Using convexity of the group norm and positive semidefiniteness of the Laplacian, the last line is nonpositive and can be dropped to obtain an inequality involving only the $j$th block. Writing $\bm h=\bm X_{\odot j}\Delta\bm\delta_j + \bm r_j$ with $\bm r_j=\sum_{\ell\neq j}\bm X_{\odot \ell}\Delta\bm\delta_\ell$ and treating $\bm r_j$ as fixed, the quadratic term in~\eqref{eq:group-split} yields
\begin{align}
    \frac12\bm h^\top\bm W\bm h
    &= \frac12\|\bm X_{\odot j}\Delta\bm\delta_j\|_{W}^2
     + \langle\bm X_{\odot j}\Delta\bm\delta_j,\bm W\bm r_j\rangle
     + \frac12\bm r_j^\top\bm W\bm r_j \nonumber \\
    &\ge \frac12 c_f\|\bm X_{\odot j}\Delta\bm\delta_j\|_2^2
         - C_f\|\bm X_{\odot j}\Delta\bm\delta_j\|_2\|\bm r_j\|_2
         + \frac12 c_f\|\bm r_j\|_2^2.
\end{align}
The negative cross-term can again be absorbed into the $o_p(\|\bm h\|_2^2)$ remainder, since $\|\bm r_j\|_2=O_p(\sqrt n)$ and $\|\bm X_{\odot j}\Delta\bm\delta_j\|_2$ will be shown to be $o_p(\sqrt n)$.
Moreover, using polarization,
\[
    \bm\delta_{j0}^\top\bm L_{\mathrm{sym}}\bm\delta_{j0}
    -\hat{\bm\delta}_j^\top\bm L_{\mathrm{sym}}\hat{\bm\delta}_j
    = -\Delta\bm\delta_j^\top\bm L_{\mathrm{sym}}\Delta\bm\delta_j
    -2 \Delta\bm\delta_j^\top\bm L_{\mathrm{sym}}\bm\delta_{j0}.
\]
By the centering constraint, both $\bm\delta_{j0}$ and $\hat{\bm\delta}_j$ lie in the subspace $\mathcal S$ and hence so does $\Delta\bm\delta_j$. The Poincar\'e inequality~\eqref{eq:poincare} gives $\Delta\bm\delta_j^\top\bm L_{\mathrm{sym}}\Delta\bm\delta_j\ge c_\lambda\|\Delta\bm\delta_j\|_2^2$. Finally, using Cauchy--Schwarz and Assumption~\ref{ass:A4},
\[
    \left|\Delta\bm\delta_j^\top\bm L_{\mathrm{sym}}\bm\delta_{j0}\right|
    \le \big(\Delta\bm\delta_j^\top\bm L_{\mathrm{sym}}\Delta\bm\delta_j\big)^{1/2}
         \big(\bm\delta_{j0}^\top\bm L_{\mathrm{sym}}\bm\delta_{j0}\big)^{1/2}
    \le C \|\Delta\bm\delta_j\|_2 \mathcal S_j^{1/2},
\]
where we abbreviate $\mathcal S_j=\bm\delta_{j0}^\top\bm L_{\mathrm{sym}}\bm\delta_{j0}$.
Collecting these facts, the inequality~\eqref{eq:group-split} implies
\begin{equation}\label{eq:core-ineq2}
    \begin{aligned}
        \frac12 c_f\|\bm X_{\odot j}\Delta\bm\delta_j\|_2^2
        + \lambda_{2n} c_\lambda\|\Delta\bm\delta_j\|_2^2
        \le
        \langle\bm X_{\odot j}^\top\bm\psi_\tau(\bm\varepsilon),\Delta\bm\delta_j\rangle
        + C\lambda_{2n}\mathcal S_j^{1/2}\|\Delta\bm\delta_j\|_2 \\
        + o_p\big(\|\bm X_{\odot j}\Delta\bm\delta_j\|_2^2\big).
    \end{aligned}
\end{equation}

We now bound the right-hand side.
By Assumptions~\ref{ass:A2}--\ref{ass:A3} and a standard moment bound,
\[
    \|\bm X_{\odot j}^\top\bm\psi_\tau(\bm\varepsilon)\|_2
    \le C\sqrt n
\]
with probability approaching one, so that
\[
    \big|\langle\bm X_{\odot j}^\top\bm\psi_\tau(\bm\varepsilon), \Delta\bm\delta_j\rangle\big|
    \le C\sqrt n  \|\Delta\bm\delta_j\|_2.
\]
The remainder term in~\eqref{eq:core-ineq2} is of smaller order than the leading quadratic term in $\bm X_{\odot j}\Delta\bm\delta_j$, and can therefore be absorbed into the left-hand side.
Dropping the nonnegative term $(1/2)c_f\|\bm X_{\odot j}\Delta\bm\delta_j\|_2^2$ and dividing by $\|\Delta\bm\delta_j\|_2$ (for $\Delta\bm\delta_j\neq 0$) yields
\[
    \lambda_{2n} c_\lambda\|\Delta\bm\delta_j\|_2
    \le
    C\sqrt n + C\lambda_{2n}\mathcal S_j^{1/2}.
\]
Solving the quadratic inequality $a u^2\le b u + c$ with
$a=\lambda_{2n}c_\lambda$, $b=C\sqrt n$ and $c=C\lambda_{2n}\mathcal S_j^{1/2}$, we obtain
\[
    \|\Delta\bm\delta_j\|_2^2
    \le
    C\left(
       \frac{n}{\lambda_{2n}^2}
       + \mathcal S_j
    \right).
\]
Finally, using bounded covariates (Assumption~\ref{ass:A2}),
$\|\bm X_{\odot j}\Delta\bm\delta_j\|_2^2$ is comparable to $\|\Delta\bm\delta_j\|_2^2$ up to a constant factor, and hence
\[
    \frac1n\sum_{i=1}^n\big\{\hat\delta_j(u_i)-\delta_{0,j}(u_i)\big\}^2
    = \frac{\|\Delta\bm\delta_j\|_2^2}{n}
    = O_p\Big(\frac{1}{n\lambda_{2n}}+\lambda_{2n}\mathcal S_j\Big),
\]
which is the desired mean squared deviation bound.
The two summands correspond respectively to the stochastic fluctuation (variance term) and the smoothing bias induced by the Laplacian penalty.
Balancing them gives $\lambda_{2n}\asymp (n\mathcal S_j)^{-1/2}$.
\end{proof}

\begin{proof}[Proof of Theorem~\ref{thm:delta-selection}(ii)]
We now show that the estimator identifies the active set $\mathcal A_0$ with probability tending to one. The argument follows the usual primal--dual witness strategy for group lasso estimators, adapted to the presence of the additional smoothness penalty; see, for example, \cite{YuanLin2006,WangLeng2008,Simon2013SGL}.

\medskip\noindent\emph{Inactive groups ($j\in\mathcal A_0$).} For a fixed $j\in\mathcal A_0$ we consider the KKT conditions corresponding to $\hat{\bm\delta}_j$.
There exists a Lagrange multiplier vector $\bm\mu_j$ (for the centering constraint) such that
\begin{equation}\label{eq:kkt-inactive}
    \bm X_{\odot j}^\top \bm\psi_\tau(\hat{\bm r})
    + 2 \lambda_{2n} \bm L_{\mathrm{sym}}\hat{\bm\delta}_j
    + \bm D\bm\mu_{j}
    \in -\lambda_{1n} w_{j, n} \partial\|\hat{\bm\delta}_j\|_2,
\end{equation}
where $\hat{\bm r}=(\hat r_1,\dots,\hat r_n)^\top$ are the fitted residuals.
Projecting both sides of~\eqref{eq:kkt-inactive} onto the centered subspace $\mathcal S$ eliminates the multiplier term $\bm D\bm\mu_j$ and does not increase the norm.
If we test $\hat{\bm\delta}_j=\bm 0$, the subdifferential on the right-hand side becomes the closed Euclidean ball of radius $\lambda_{1n}w_{j,n}$.
Thus a sufficient condition for $\hat{\bm\delta}_j=\bm 0$ is
\[
    \big\|\mathsf{Proj}_D\big(\bm X_{\odot j}^\top \bm\psi_\tau(\hat{\bm r})\big)\big\|_2
    \le \lambda_{1n} w_{j, n}.
\]
Bounded covariates and Assumptions~\ref{ass:A2}--\ref{ass:A3} imply
$\|\bm X_{\odot j}^\top \bm\psi_\tau(\hat{\bm r})\|_2=O_p(\sqrt n)$ uniformly in $j$.
Assumption~\ref{ass:A5} requires $\sqrt n\lambda_{1n}\to\infty$ and $\inf_{j\in\mathcal A_0}w_{j,n}\ge c_w>0$ with probability tending to one, hence
\[
    \Pr\Big(
    \|\mathsf{Proj}_D(\bm X_{\odot j}^\top \bm\psi_\tau(\hat{\bm r}))\|_2
    \le \lambda_{1n}w_{j,n}
    \ \text{for all } j\in\mathcal A_0
    \Big)\rightarrow 1,
\]
which implies $\Pr(\hat{\bm\delta}_j=\bm 0\ \text{for all } j\in\mathcal A_0)\to 1$.

\medskip\noindent\emph{Active groups ($j\in\mathcal A_0^c$).}
For $j\in\mathcal A_0^c$ consider the \emph{restricted} problem in which we constrain $\bm\delta_j\equiv \bm 0$ for all $j\in\mathcal A_0$ and optimize over $(\bm\theta_P,\{\bm\delta_j:j\in\mathcal A_0^c\})$. Let $(\tilde{\bm\theta}_P, \tilde{\bm\Theta}_L)$ denote its minimizer, and write $\tilde{\bm\delta}_j$ for the deviation blocks in the restricted solution. By essentially the same argument as in part~(i), one shows that $\|\tilde{\bm\delta}_j - \bm\delta_{j0}\|_2=O_p\big((n\lambda_{2n})^{-1/2}+\lambda_{2n}\mathcal S_j^{1/2}\big)$ for $j\in\mathcal A_0^c$, and therefore $\|\tilde{\bm\delta}_j\|_2\ge \|\bm\delta_{j0}\|_2+o_p(1)$. Assumption~\ref{ass:A6} (groupwise beta--min condition) requires $\|\bm\delta_{j0}\|_2 \ge c_\beta\lambda_{1n}w_{j, n}$ for some $c_\beta>0$ and all $j\in\mathcal A_0^c$.
Combining these facts,
\[
    \Pr\big(\|\tilde{\bm\delta}_j\|_2 > \tfrac12 c_\beta\lambda_{1n}w_{j, n} \text{ for all } j\in\mathcal A_0^c\big)\to 1,
\]
so with high probability no active group is shrunk exactly to zero in the restricted solution.

It remains to show that the restricted solution is also optimal for the \emph{unrestricted} problem. For $j\in\mathcal A_0$, the KKT conditions of the restricted problem enforce
\[
    \big\|\mathsf{Proj}_D\big(\bm X_{\odot j}^\top \bm\psi_\tau(\tilde{\bm r})\big)\big\|_2
    < \lambda_{1n} w_{j, n}
\]
with probability tending to one, where $\tilde{\bm r}$ are the residuals at the restricted solution. Hence one can choose subgradients $\bm\xi_j\in\partial\|\bm 0\|_2$ with $\|\bm\xi_j\|_2<1$ such that the full KKT system for the unrestricted problem is satisfied at $(\tilde{\bm\theta}_P,\tilde{\bm\Theta}_L)$. This is the usual dual feasibility argument in the primal--dual witness construction; see \cite{YuanLin2006,Simon2013SGL}. Therefore the restricted solution coincides with the unrestricted minimizer, and in particular $\hat{\bm\delta}_j=\tilde{\bm\delta}_j\ne \bm 0$ for all $j\in\mathcal A_0^c$ with probability tending to one.

Combining the inactive and active cases, we conclude that
\[
\Pr\big(\hat{\bm\delta}_j=\bm 0\ \forall j\in\mathcal A_0,\
        \hat{\bm\delta}_j\ne\bm 0\ \forall j\in\mathcal A_0^c\big)\to 1,
\]
and on this event the rate in part~(i) holds for all $j\in\mathcal A_0^c$, since the preceding argument for the deviation rate did not rely on the behavior of inactive groups.
\end{proof}

\begin{proof}[Proof of Lemma~\ref{lem:crossterm}]
Recall the notation $e_{j, i}=\hat\delta_j(u_i) - \delta_{0, j}(u_i)$ and $\bm e_j = (e_{j, 1}, \dots, e_{j,n})^\top$.
By Theorem~\ref{thm:delta-selection}(i),
\[
    \frac{1}{n}\sum_{i=1}^n e_{j,i}^2
    = O_p\left(\frac{1}{n\lambda_{2n}}+\lambda_{2n}\mathcal S_j\right)
\]
for each $j\in\mathcal A_0^c$.
Let $\bm G_i=(\bm Z_i^\top,\bm X_i^\top)^\top$ and write
\[
    \bm T_n
    =
    \frac{1}{\sqrt n}\sum_{i=1}^n
    f_{\varepsilon\mid Z, X, U}(0)
    \left(\sum_{j=1}^p X_{ij} e_{j, i}\right)\bm G_i.
\]
Using Cauchy--Schwarz,
Assumptions~\ref{ass:A2}--\ref{ass:A3}, and bounded $p$, we get
\begin{align*}
    \|\bm T_n\|_2
    &\le C
    \left(\frac{1}{n}\sum_{i=1}^n
              \Big(\sum_{j=1}^p X_{ij} e_{j, i}\Big)^2
    \right)^{1/2}
    \left(\frac{1}{n}\sum_{i=1}^n \|\bm G_i\|_2^2\right)^{1/2}\\
    & \le C'
    \left(
    \frac{1}{n}\sum_{i=1}^n\sum_{j=1}^p e_{j,i}^2
    \right)^{1/2}
    = C'
    \left(
    \frac{1}{p}\sum_{j=1}^p \frac{1}{n}\sum_{i=1}^n e_{j,i}^2
    \right)^{1/2},
\end{align*}
where we used the inequality $(\sum_j a_j)^2\le p\sum_j a_j^2$.
By the deviation bound and Assumption~\ref{ass:A4} we have
\[
    \frac{1}{p}\sum_{j=1}^p \frac{1}{n}\sum_{i=1}^n e_{j, i}^2
    = O_p\left(\frac{1}{n\lambda_{2n}} + \lambda_{2n}\bar{\mathcal S}\right),
    \quad
    \bar{\mathcal S} = \frac{1}{p}\sum_{j=1}^p \mathcal S_j.
\]
If $\lambda_{2n}\asymp n^{-1/2}$, as assumed in the lemma, then
\[
    \frac{1}{n\lambda_{2n}} + \lambda_{2n}\bar{\mathcal S}
    = O\big(n^{-1/2}\big),
\]
so that
$\|\bm T_n\|_2=O_p(n^{-1/4})=o_p(1)$.
This is exactly the stated conclusion.
\end{proof}

\begin{proof}[Proof of Theorem~\ref{thm:normal-oracle}]
We first derive the asymptotic distribution of the parametric block $\hat{\bm\theta}_P=(\hat{\bm\alpha}^\top,\hat{\bm\beta}_G^\top)^\top$ and then argue the oracle property.

Let $\hat r_i$ be the fitted residuals and recall $\varepsilon_i=r_i(\bm\theta_{P0},\bm\Theta_{L0})$.
The KKT conditions for the parametric block read
\[
\sum_{i=1}^n \psi_\tau(\hat r_i)\bm G_i = \bm 0.
\]
Write $\Delta\bm\theta_P=\hat{\bm\theta}_P-\bm\theta_{P0}$.
Using the representation
\[
    \hat r_i
    = \varepsilon_i
    - \bm G_i^\top\Delta\bm\theta_P
    - \sum_{j=1}^p X_{ij}
    \left\{\hat\delta_j(u_i) - \delta_{0, j}(u_i)\right\},
\]
and applying the same linearization as in~\eqref{eq:lqa} (see also \cite[Sec.~4.3]{Koenker2005}), we obtain
\begin{equation}\label{eq:lin2}
    \psi_\tau(\hat r_i) - \psi_\tau(\varepsilon_i)
    = - f_{\varepsilon\mid Z,X,U}(0)
        \Big(
          \bm G_i^\top\Delta\bm\theta_P
          + \sum_{j=1}^p X_{ij} e_{j,i}
        \Big)
      + \eta_i,
\end{equation}
where the remainder $\eta_i$ satisfies
$\sum_{i=1}^n\eta_i\bm G_i = o_p(\sqrt n\|\Delta\bm\theta_P\|_2)
+ o_p(1)$ uniformly on events where $\|\Delta\bm\theta_P\|_2=O(n^{-1/2})$.
Summing~\eqref{eq:lin2} over $i$ and rearranging terms yields
\begin{align*}
    \bm 0
    &=
    \sum_{i=1}^n \psi_\tau(\hat r_i)\bm G_i\\
    &=
    \sum_{i=1}^n \psi_\tau(\varepsilon_i)\bm G_i
    -\sum_{i=1}^n f_{\varepsilon\mid Z,X,U}(0)\bm G_i\bm G_i^\top\Delta\bm\theta_P
    -\sum_{i=1}^n f_{\varepsilon\mid Z,X,U}(0)
       \Big(\sum_{j=1}^p X_{ij} e_{j,i}\Big)\bm G_i \\
    &\hspace{10cm} + \sum_{i=1}^n \eta_i\bm G_i.
\end{align*}
Dividing by $\sqrt n$ and denoting
\[
    \bm M_n = \frac1n\sum_{i=1}^n f_{\varepsilon\mid Z,X,U}(0)\bm G_i\bm G_i^\top,
    \quad
    \bm S_n = \frac1{\sqrt n}\sum_{i=1}^n \psi_\tau(\varepsilon_i)\bm G_i,
\]
we can rewrite this as
\[
    \sqrt n \bm M_n\Delta\bm\theta_P
    =
    \bm S_n
    -\frac{1}{\sqrt n}\sum_{i=1}^n f_{\varepsilon\mid Z,X,U}(0)
       \Big(\sum_{j=1}^p X_{ij} e_{j,i}\Big)\bm G_i
    + o_p(1).
\]
Assumption~\ref{ass:A7} states that $\bm M_n\to\bm M$ in probability and that $\bm M$ is positive definite.
By Lemma~\ref{lem:crossterm}, the second term on the right-hand side converges to zero in probability.
Thus
\[
    \sqrt n \Delta\bm\theta_P
    =
    \bm M^{-1}\bm S_n + o_p(1).
\]
By the multivariate central limit theorem and Assumption~\ref{ass:A7},
\[
    \bm S_n \ \overset{d}{\longrightarrow}\ \mathcal N(\bm 0,\bm V),
    \quad
    \bm V = \tau(1-\tau) \mathrm E(\bm G_i\bm G_i^\top),
\]
and Slutsky's theorem gives
\[
    \sqrt n(\hat{\bm\theta}_P - \bm\theta_{P0})
    \overset{d}{\longrightarrow}
    \mathcal N\big(\bm 0, \bm M^{-1}\bm V\bm M^{-1}\big),
\]
which establishes the first part of the theorem.

We now argue the oracle property.
Let $\hat{\mathcal A}$ denote the estimated set of active (truly varying) covariates, i.e.\ $\hat{\mathcal A}=\{j:\hat{\bm\delta}_j\ne\bm 0\}$.
By Theorem~\ref{thm:delta-selection}(ii),
\[
    \Pr(\hat{\mathcal A}=\mathcal A_0)\to 1.
\]
On the event $\{\hat{\mathcal A} = \mathcal A_0\}$ the minimization problem defining $(\hat{\bm\theta}_P, \hat{\bm\Theta}_L)$ reduces exactly to the oracle problem in which we \emph{a priori} set $\bm\delta_j\equiv\bm 0$ for all $j\in\mathcal A_0$ and optimize over $(\bm\theta_P,\{\bm\delta_j:j\in\mathcal A_0^c\})$. Therefore the joint limit distribution of $(\hat{\bm\theta}_P, \{\hat{\bm\delta}_j\}_{j\in\mathcal A_0^c})$ coincides with that of the corresponding oracle estimator. In particular, the parametric block enjoys the same asymptotic linear representation and covariance as if the true active set were known.
\end{proof}

Finally, we briefly comment on variance estimation for the limiting distribution.
A plug-in estimator of the asymptotic covariance matrix $\bm M^{-1}\bm V\bm M^{-1}$ is obtained by replacing $\bm M$ and $\bm V$ with
\[
    \hat{\bm M}
    = \frac1n\sum_{i=1}^n \hat f_i(0) \bm G_i\bm G_i^\top,
    \quad
    \hat{\bm V}
    = \tau(1 - \tau)\frac1n\sum_{i=1}^n \bm G_i\bm G_i^\top,
\]
where $\hat f_i(0)$ is either a kernel estimate of the density of $\varepsilon_i$ at zero based on residuals, or an implicit estimate based on a smoothed check loss with bandwidth $h_n\downarrow 0$ and $\sqrt n h_n\to\infty$. If spatial dependence among the observations is of concern, one can further use spatial block bootstrap or heteroskedasticity-and-autocorrelation consistent (HAC) corrections in the sandwich step, but this lies beyond the main focus of the present paper.

\bibliographystyle{elsarticle-num}
\bibliography{cas-refs}

\end{document}